\documentclass[10pt]{amsart}%
\usepackage{a4wide,color}
\usepackage{amsmath,amssymb}
\usepackage[active]{srcltx}
\usepackage{enumerate}
\usepackage{amsmath}
\usepackage{mathrsfs}
\usepackage{amsfonts}
\usepackage{amssymb}
\usepackage{graphicx}
\usepackage{ulem}
\usepackage{geometry}%
\providecommand{\U}[1]{\protect\rule{.1in}{.1in}}
\allowdisplaybreaks
\let\pa\partial

\newtheorem{theorem}{Theorem}
\newtheorem{lemma}[theorem]{Lemma}
\newtheorem{proposition}[theorem]{Proposition}

\let\ga=\gamma
\let\de=\delta

\let\la=\lambda

\let\pa=\partial

\let\om=\omega

\begin{document}
\title[Fokker-Planck equation]{Explicit Structure of the Fokker-Planck Equation with potential}
\author{Yu-Chu Lin}
\address{Yu-Chu Lin, Department of Mathematics, National Cheng Kung University, Tainan, Taiwan}
\email{yuchu@mail.ncku.edu.tw }
\author{Haitao Wang}
\address{Haitao Wang, Institute of Natural Sciences and School of Mathematical
Sciences, Shanghai Jiao Tong University, Shanghai, China}
\email{haitallica@sjtu.edu.cn}
\author{Kung-Chien Wu}
\address{Kung-Chien Wu, Department of Mathematics, National Cheng Kung University,
Tainan, Taiwan and National Center for Theoretical Sciences, National Taiwan
University, Taipei, Taiwan}
\email{kungchienwu@gmail.com}
\thanks{The first author is supported by the Ministry of Science and Technology under
the grant MOST 105-2115-M-006-002-. The second author is sponsored by Shanghai
Sailing Program(18YF1411800) and Shanghai Jiao Tong University(WF220441907).
The third author is supported by the Ministry of Science and Technology under
the grant 104-2628-M-006-003-MY4 and National Center for Theoretical Sciences.}

\begin{abstract}
We study the pointwise (in the space and time variables) behavior of the
Fokker-Planck Equation with potential. An explicit description of the
solution is given, including the large time behavior, initial layer and
spatially asymptotic behavior. Moreover, it is shown that the structure of the
solution sensitively depends on the potential function.

\end{abstract}
\maketitle





\section{Introduction}

\subsection{The Models}

The Fokker-Planck equations arise in many areas of sciences, including
probability, statistical physics, plasma physics, gas and stellar dynamics.
The term \textquotedblleft Fokker-Planck\textquotedblright\ is widely used to
represent various diffusion processes (Brownian motion).

In this paper, we study the kinetic Fokker-Planck equation with potential in $\mathbb{R}^{3}$. It reads
\begin{equation}
\left\{
\begin{array}
[c]{l}%
\displaystyle\pa_{t}F+v\cdot\nabla_{x}F=\nabla_{v}\cdot\left[  \nabla
_{v}F+(\nabla_{v}\Phi)F\right]  \,,\quad x,v\in\mathbb{R}^{3},\ t>0,\\
\\
\displaystyle F(0,x,v)=F_{0}(x,v)\,,
\end{array}
\right.  \label{in.1.a}%
\end{equation}
where the potential $\Phi(v)$ is of the form
\[
\Phi=\frac{1}{\ga}\left\langle v\right\rangle ^{\ga}+\Phi_{0}\,,\ \ga>0,
\]
for some constant $\Phi_{0}.$ We define
\[
\mathcal{M}(v)=e^{-\Phi(v)}\,,
\]
with $\Phi_{0}\in\mathbb{R}\ $such that $\mathcal{M}$ is a probability
measure. It is easy to see that $\mathcal{M}$ is a steady state to the
Fokker-Planck equation (\ref{in.1.a}). Thus it is natural to study the
fluctuation of the Fokker-Planck equation (\ref{in.1.a}) around $\mathcal{M}%
(v)$, with the standard perturbation $f(t,x,v)$ to $\mathcal{M}$ as
\[
F=\mathcal{M}+\mathcal{M}^{1/2}f\,.
\]
The Fokker-Planck equation for $f(t,x,v)=\mathbb{G}^{t}f_{0}$ now takes the
form
\begin{equation}
\left\{
\begin{array}
[c]{l}%
\displaystyle\pa_{t}f+v\cdot\nabla_{x}f=\Delta_{v}f-\frac{1}{4}|v|^{2}%
\left\langle v\right\rangle ^{2\ga-4}f+\left(  \frac{3}{2}\left\langle
v\right\rangle ^{\ga-2}+\frac{\ga-2}{2}|v|^{2}\left\langle v\right\rangle
^{\ga-4}\right)  f=Lf\,,\\
\\
f(0,x,v)=f_{0}(x,v)\,,\ \ \ \left(  x,v\right)  \in\mathbb{R}^{3}%
\times\mathbb{R}^{3}\,.
\end{array}
\right.  \label{in.1.c}%
\end{equation}
Here $\mathbb{G}^{t}$ is the solution operator of the Fokker-Planck equation
(\ref{in.1.c}). It is obvious that $L$ is a non-positive self-adjoint operator
on $L_{v}^{2}.$ More precisely, its Dirichlet form is given by
\[
\left\langle Lf,f\right\rangle _{v}=-\int_{\mathbb{R}^{3}}\left\vert
\nabla_{v}f+\frac{\nabla\Phi}{2}f\right\vert ^{2}dv=-\int_{\mathbb{R}^{3}%
}\left\vert \nabla_{v}\left(  \frac{f}{\sqrt{\mathcal{M}}}\right)  \right\vert
^{2}\mathcal{M}dv.
\]
Therefore, the null space of $L$ is given by
\[
Ker(L)=\hbox{span}\left\{  E_{D}\right\}  \,,
\]
where $E_{D}=\sqrt{\mathcal{M}}$. Based on this property, we can introduce the
macro-micro decomposition as follows: the macro projection $\mathrm{P}_{0}$ is
the orthogonal projection with respect to the $L_{v}^{2}$ inner product onto
$\mathrm{Ker}(L)$, and the micro projection $\mathrm{P}_{1}\equiv
\mathrm{Id}-\mathrm{P}_{0}$.

\subsection{Main theorem}

Before the presentation of the main theorem, let us define some notation in
this paper. We denote $\left\langle v\right\rangle ^{s}=(1+|v|^{2})^{s/2}$,
$s\in{\mathbb{R}}$. For the microscopic variable $v$, we denote
\[
|f|_{L_{v}^{2}}=\Big(\int_{{\mathbb{R}}^{3}}|f|^{2}dv\Big)^{1/2},
\]
and the weighted norms $|\cdot|_{L_{v}^{2}(m)}$ and $|\cdot|_{L_{\theta}^{2}}$
can be defined by
\[
|f|_{L^{2}(m)}=\Big(\int_{\mathbb{R}^{3}}|f|^{2}mdv\Big)^{1/2}\,,\quad
|f|_{L_{\theta}^{2}}=\Big(\int_{\mathbb{R}^{3}}\left\langle v\right\rangle
^{2\theta}|f|^{2}dv\Big)^{1/2},
\]
respectively, where $m=m\left(  t,x,v\right)  $ is a weight function. The
$L_{v}^{2}$ inner product in ${\mathbb{R}}^{3}$ will be denoted by
$\left\langle \cdot,\cdot\right\rangle _{v}$,
\[
\left\langle f,g\right\rangle _{v}=\int_{{\mathbb{R}}^{3}}f(v)\overline
{g(v)}dv.
\]
For the space variable $x$, we have the similar notation. In fact, $L_{x}^{2}$
is the classical Hilbert space with norm
\[
|f|_{L_{x}^{2}}=\Big(\int_{\mathbb{R}^{3}}|f|^{2}dx\Big)^{1/2}\,.
\]
We denote the supremum norm as
\[
|f|_{L_{x}^{\infty}}=\sup_{x\in{\mathbb{R}^{3}}}|f(x)|\,.
\]
The standard inner product in $\mathbb{R}^{3}$ will be denoted by
$(\cdot,\cdot)$. For the Fokker-Planck equation, the natural space in the $v$
variable is equipped with the norm $|\cdot|_{L_{\sigma}^{2}}$, which is
defined as
\[
|f|_{L_{\sigma}^{2}}^{2}=|\big<v\big>^{\ga-1}f|_{L_{v}^{2}}^{2}+|\nabla
_{v}f|_{L_{v}^{2}}^{2}\,,
\]
and the corresponding weighted norms are defined as
\[
|f|_{L_{\sigma}^{2}(m)}^{2}=|\big<v\big>^{\ga-1}f|_{L_{v}^{2}(m)}^{2}%
+|\nabla_{v}f|_{L_{v}^{2}(m)}^{2}\,,\quad|f|_{L_{\sigma,\theta}^{2}}%
^{2}=|\big<v\big>^{\ga-1}f|_{L_{\theta}^{2}}^{2}+|\nabla_{v}f|_{L_{\theta}%
^{2}}^{2}\,.
\]
Moreover, we define
\[
\Vert f\Vert_{L^{2}}^{2}=\int_{{\mathbb{R}}^{3}}|f|_{L_{v}^{2}}^{2}%
dx\,,\quad\Vert f\Vert_{L_{\sigma}^{2}}^{2}=\int_{{\mathbb{R}}^{3}%
}|f|_{L_{\sigma}^{2}}^{2}dx\,,
\]
and
\[
\Vert f\Vert_{L_{x}^{\infty}L_{v}^{2}}=\sup_{x\in{\mathbb{R}^{3}}}%
|f|_{L_{v}^{2}}\,,\quad\Vert f\Vert_{L_{x}^{1}L_{v}^{2}}=\int_{{\mathbb{R}%
^{3}}}|f|_{L_{v}^{2}}dx\,.
\]
Finally, we define the high order Sobolev norm in $x$ variable: let
$k\in{\mathbb{N}}$ and let $\alpha$ be any multi-index,
\[
\left\Vert f\right\Vert _{H_{x}^{k}L_{v}^{2}}:=\sum_{|\alpha|\leq k}\left\Vert
\partial_{x}^{\alpha}f\right\Vert _{L^{2}}\,.
\]
The weighted spaces in the $(x,v)$-variable can be defined in a similar way.

For multi-indices $\alpha$, $\beta_{j}(j=1,\dots,s)\in\mathbb{N}_{0}^{3}$ with
$\alpha=\sum\limits_{j=1}^{s}\beta_{j}$, we denote the multinomial
coefficients by
\[
\binom{\alpha}{\beta_{1}\,\beta_{2}\,\dots\,\beta_{s}}=\prod_{k=1}^{3}%
\frac{\alpha_{k}!}{\prod_{j=1}^{s}(\beta_{j})_{k}!}\,.
\]

The domain decomposition plays an essential role in our analysis, hence we
define a cut-off function $\chi:{\mathbb{R}}\rightarrow{\mathbb{R}}$, which is
a smooth non-increasing function, $\chi(s)=1$ for $s\leq1$, $\chi(s)=0$ for
$s\geq2$ and $0\leq\chi\leq1$. Moreover, we define $\chi_{R}(s)=\chi(s/R)$.

For simplicity of notation, hereafter, we abbreviate \textquotedblleft{\ $\leq
C$} \textquotedblright\ to \textquotedblleft{\ $\lesssim$ }\textquotedblright,
where $C$ is a positive constant depending only upon fixed numbers.

Here is the precise description of our main results (combining theorem \ref{time-like region for gamma larger or equal than 1}, theorem \ref{time-like region for gamma less than 1}, theorem \ref{space-like region for gamma greater or equal to 3/2} and theorem \ref{space-like region for 0<gamma<3/2}):

\begin{theorem}
Let $f$ be a solution to the Fokker-Planck equation \eqref{in.1.c} with
initial data compactly supported in the $x$ variable and bounded in $L_{v}%
^{2}$ (we need some exponential weight for $0<\gamma<3/2$) space
\[
f_{0}(x,v)\equiv0\text{ for }\left\vert x\right\vert \geq1.
\]
There exists a positive constant $M$ such that the following hold:

\begin{enumerate}
\item As $\gamma\geq3/2$, there exists a positive constant $C$ such that the
solution $f$ satisfies

\begin{enumerate}
\item For $\left\langle x\right\rangle \leq2Mt$,
\[
\left\vert f(t,x)\right\vert _{L_{v}^{2}}\lesssim\left[  (1+t^{-9/4})e^{-Ct}+
(1+t)^{-3/2}e^{-C\frac{|x|^{2}}{t+1}}\right]  \Vert f_{0}\Vert_{L_{x}^{\infty
}L_{v}^{2}}\,.
\]

\item For $\left\langle x\right\rangle \geq2Mt$,
\[
\left\vert f(t,x)\right\vert _{L_{v}^{2}}\lesssim(1+t^{-9/4})e^{-C\left(
\left\langle x\right\rangle +t\right)  } \Vert f_{0}\Vert_{L_{x}^{\infty}%
L_{v}^{2}}\,.
\]

\end{enumerate}

\item As $1\leq\gamma<3/2$, for any given positive integer $N$ and any
sufficiently small $\alpha>0$, there exists a positive constant $C$ such that
the solution $f$ satisfies

\begin{enumerate}
\item For $\left\langle x\right\rangle \leq2Mt$,
\[
\left\vert f(t,x)\right\vert _{L_{v}^{2}}\lesssim\left[  (1+t^{-9/4}%
)e^{-Ct}+(1+t)^{-3/2}\Big(1+\frac{|x|^{2}}{1+t}\Big)^{-N} \right]  \Vert
f_{0}\Vert_{L_{x}^{\infty}L_{v}^{2}}\,.
\]

\item For $\left\langle x\right\rangle \geq2Mt$,
\[
\left\vert f(t,x)\right\vert _{L_{v}^{2}}(1+t^{-9/4})e^{-C(\left\langle
x\right\rangle +t)^{\frac{\gamma}{3-\gamma}}}\Vert f_{0}\Vert_{L^{2}%
(e^{4\alpha\left\langle v\right\rangle ^{\gamma}})}\,.
\]

\end{enumerate}

\item As $0<\gamma<1$, for any sufficiently small $\alpha>0$, there exists a
positive constant $C$ such that the solution $f$ satisfies

\begin{enumerate}
\item For $\left\langle x\right\rangle \leq2Mt$,
\[
\left\vert f(t,x)\right\vert _{L_{v}^{2}}\lesssim\left[  (1+t^{-9/4}%
)e^{-Ct^{\frac{\gamma}{2-\gamma}}}+(1+t)^{-3/2}\right]  \left\Vert
f_{0}\right\Vert _{L^{2}(e^{4\alpha\left\langle v\right\rangle ^{\gamma}}%
)}\,.
\]

\item For $\left\langle x\right\rangle \geq2Mt$,
\[
\left\vert f(t,x)\right\vert _{L_{v}^{2}}\lesssim(1+t^{-9/4}%
)e^{-C(\left\langle x\right\rangle +t)^{\frac{\gamma}{3-\gamma}}}\Vert
f_{0}\Vert_{L^{2}(e^{4\alpha\left\langle v\right\rangle ^{\gamma}})}\,.
\]

\end{enumerate}
\end{enumerate}
\end{theorem}

\subsection{Review of previous works and significant points of the paper}

The study of the Fokker-Planck equation can be traced back to 1930's. When the
potential $\Phi=0$, the equation (\ref{in.1.a}) is known as the
Kolmogorov-Fokker-Planck equation.
In 1934 Kolmogorov \cite{[Kolmogorov]} derived the Green function for the
whole space problem. The explicit formula surprisingly showed that the
solution becomes smooth in the $t,x,v$ variables when $t>0$ immediately.

Later the regularization effect has been investigated further and been
recovered by some more general and robust methods. For example, it is known
that the Fokker-Planck operator $-v\cdot\nabla_{x}+\Delta_{v}$ is a
hypoelliptic operator. So one can apply H\"{o}rmander's commutator
\cite{[Hormander]} to the linear Fokker-Planck operator to obtain that
diffusion in $v$ together with the transport term $v\cdot\nabla_{x}$ has a
regularizing effect on solutions not only in $v$ but also in $t$ and $x$. It
can also be obtained through the functional method, see \cite{[Herau],Villani}%
. On the other hand, the Fokker-Planck operator is also known as a
hypocoercive operator, which concerns the rate of convergence to equilibrium.
Indeed, the trend to equilibria with a certain rate has been investigated in
many papers (cf. \cite{[Desvillettes],[Duan],[Herau],[Herau-Nier], [MisMou],
mouNeu}) for the close to Maxwellian regime in the whole space or in the
periodic box.

Let us point out the recent important results constructed by Mouhot and
Mischler \cite{[MisMou]}. They developed an abstract method for deriving decay
estimates of the semigroup associated to non-symmetric operators in Banach
spaces. Applying this method to the kinetic Fokker-Planck equation in the
torus with potential in the close to equilibrium setting, they obtained
spectral gap estimates for the associated semigroup in various norms,
including Lebesgue norms, negative Sobolev norms, and the
Monge-Kantorovich-Wasserstein distance $W_{1}$.

In this paper, we study the Fokker-Planck equation with potential
in the close to equilibrium setting. In the literature, this kind of problem
basically focuses on the rate of convergence to equilibrium (see the reference
listed above). Instead, in this paper we supply an explicit description of the
solution in the sense of pointwise estimate. It turns out the structure of the
solution sensitively depends on the potential function. Let us illustrate the
novelties of the paper:

\begin{itemize}
\item We obtain the global picture of the solution, which consists of three
parts: the time-like region (large time behavior), the space-like region
(spatially asymptotic behavior) and the small time region (the evolution of
initial singularity).

\begin{enumerate}
\item In the time-like region, we have distinctly different descriptions
according to potential functions. For $\gamma\geq1$, thanks to the spectrum
analysis,
we have a pointwise fluid structure, which is more precise than previous
results. The leading term of the wave propagation has been recognized. More
specifically, for $\gamma\geq3/2$ the leading term is a diffusion wave with
heat kernel type, while for $1\leq\gamma<3/2$ the diffusion wave is of
algebraic type. By contrast, the spectral information is missing for
$0<\gamma<1$ due to the weak damping for large velocity, which leads to the
unavailability of pointwise structure. Nevertheless, we can apply Kawashima's
argument \cite{[Kawashima],[Strain]} to get a uniform time decay rate.


\item Concerning the space-like region, we have exponential decay for
$\gamma\geq3/2$ and sub-exponential decay for $0<\gamma<3/2$. The results are
consistent with the wave behaviors inside the time-like region for different
$\gamma$'s respectively. To our knowledge, this is the first result for the
asymptotic behavior of the Fokker-Planck equation with potential.

\item Owing to the regularization effect, the initial singularity is
eliminated instantaneously.


\end{enumerate}

\medskip

\item The regularization estimate is a key ingredient of this paper (see Lemma
\ref{regularization} and Lemma \ref{second-der}), which enables us to obtain
the pointwise estimate without regularity assumptions on the initial
condition.
In the literature, the regularization estimates for the kinetic Fokker-Planck
equation and Landau equation have been proved for various purposes, see for
instance \cite{[Herau]}, \cite{[MisMou]}, \cite{Villani} (Appendix A.21.2) for
the Fokker-Planck case and \cite{[CTK]} for the Landau case. {The
above-mentioned regularization estimates are sufficient for studying the time
decay of the solution. However, to gain understanding of the spatially
asymptotic behavior, one needs to analyze the solution in some appropriate
weighted spaces. Taking this into account, we construct the regularization
estimates in suitable weighted spaces. The calculation of the estimates is
interesting and more sophisticated than before. Moreover, this type of
regularization estimate is itself new.} \medskip

\item {The pointwise estimate of the solution in the space-like region is
constructed by the weighted energy estimate. The time-dependent weight
functions are chosen according to different confinement potentials. For
$\gamma\geq3/2$, from estimate in the time-like region, the solution decays
exponentially along the wave cone, i.e., $|x|=Mt$, suggesting the exponential
decay at the spatial infinity. It turns out that a simple weight function is
satisfactory (see Proposition \ref{weig_1}). However, when $0<\gamma<3/2$, we
notice that in \eqref{in.1.c} the exponent of damping coefficient($\sim
\left\langle v\right\rangle ^{2(\gamma-1)}$) is less than 1. From the scaling
of transport equation, we cannot expect exponential decay in the spatial
variable. In fact, motivating by the transport equation with weak damping, we
devise appropriate weight functions, introduce a refined space-velocity domain
decomposition and eventually show the sub-exponential decay for $0<\gamma<3/2$
(Proposition \ref{weig_2}).} \medskip

\item We believe that our idea in this paper can have potential applications
in other important kinetic equations, such as the Landau equation or Boltzmann
equation without angular cutoff. In fact, these projects are in progress.
\end{itemize}

{To the best of our knowledge, the first pointwise result of the kinetic type
equation is the Boltzmann equation for hard sphere
\cite{[LiuYu],[LiuYu2],[LiuYu1]}; the authors have established important
results regarding the pointwise behavior of the Green function and completed
the nonlinear problem. Later, the result was generalized to the Boltzmann
equation with cutoff hard potentials \cite{[LeeLiuYu]}.
Very recently, the authors of the current paper extend the pointwise result to
more general potentials, the range $-2<\gamma<1$, and obtain an explicit
relation between the decay rate and velocity weight assumption
\cite{[LinWangWu]}. Let us point out some similarities and differences between
the Fokker-Planck equation with potential and the Boltzmann equation
with hard sphere or cutoff hard potentials.}

\begin{itemize}
\item The solutions of both equations in large time are dominated by the fluid
part. For the Fokker-Planck with $\gamma\geq1$ and for the Boltzmann with hard
sphere or hard potentials with cutoff, the fluid parts are characterized by
diffusion waves. To extract them, both need the long wave-short wave
decomposition. However the wave structures of them are quite different. For
the Boltzmann equation, there are diffusion waves propagating with different
speeds: one with the background speed of the global Maxwellian while the other
with the superposed speed of the background speed and the sound speed. In
comparison, there is only one diffusion wave for the Fokker-Planck equation.
The fluid behavior can be seen formally from the Chapman-Enskog expansion,
which indicates that the macroscopic part (the fluid part) of the solution
satisfies the viscous system of conservation laws\textbf{.} For the Boltzmann
equation there are conservation laws of mass, momentum and energy, while the
Fokker-Planck equation only preserves the mass, explaining the difference of
their wave structures.
\medskip

\item Since the leading term of the solution in large time is the fluid part
and it essentially has finite propagation speed, the solution in the
space-like region,\ compared to the leading part, should be much smaller. In
fact it is shown that the asymptotic behaviors exponentially or
sub-exponentially decay. This is similar to the solution of the Boltzmann
equation outside the finite Mach number region. \medskip

\item The regularization mechanism of the Fokker-Planck equation is distinct
from that of the Boltzmann equation. For the Boltzmann equation, the initial
singularity will be preserved (although decays in time very fast), one has to
single them out. Since the singular waves satisfy a damped transport equation,
there is an explicit solution formula, from which the pointwise structure can
be deduced. Then the regularity of the resulting remainder part comes from the
compact part of the collision operator (see the Mixture Lemma in
\cite{[LeeLiuYu]}, \cite{[LiuYu]} and \cite{[LiuYu2]}). By contrast, for the
Fokker-Planck equation the regularity comes from the combined effect of
ellipticity in the velocity variable $v$ and the transport term (see Lemma
\ref{regularization} and Lemma \ref{second-der}). The initial singularities
have been identified. However, there is no explicit formula for singular
waves. Instead, they are accurately estimated by suitable weighted energy estimates.
\end{itemize}

\subsection{Method of proof and plan of the paper}

The main idea of this paper is to combine the long wave-short wave
decomposition, the wave-remainder decomposition, the weighted energy estimate
and the regularization estimate together to analyze the solution. The long
wave-short wave decomposition, based on the Fourier transform, gives the fluid
structure of the solution. The wave-remainder decomposition is used for
extracting the initial singularity. The weighted energy estimate is used for
the pointwise estimate of solution inside the space-like region, where the
regularization estimate is also used. We explain the idea in more detail as below.

In the time-like region (inside the region $|x|\leq Mt$ for some $M$), the
solution is dominated by the fluid part, which is contained in the long wave
part. In order to obtain its estimate, we devise different methods for
$\gamma\geq1$ and $0<\gamma<1$ respectively. For $\gamma\geq1$, taking
advantage of the spectrum information of the Fokker-Planck operator
\cite{[LuoYu]} (in fact, the paper \cite{[LuoYu]} only studies the case
$\gamma=2$ and we can extend it to the case $\gamma\geq1$), the complex
analytic or Fourier multiplier techniques can be applied to obtain pointwise
structure of the fluid part. However, for $0<\gamma<1$, the spectrum
information is missing due to the weak damping for large velocity. Instead, we
use Kawashima's argument \cite{[Kawashima]} to get the optimal decay only in
time. It is shown that the $L^{2}$ norm of the short wave exponentially decays
in time for $\gamma\geq1$ essentially due to the spectrum gap, while it decays
only algebraically for $0<\gamma<1$ if imposing certain velocity weight on
initial data.

We use the wave-remainder decomposition to extract the possible initial
singularity in the short wave. This decomposition is based on a Picard-type
iteration. The first several terms in the iteration contain the most singular
part of the solution, and they are the so-called wave part.
By functional methods, we prove the iteration equation has a regularization
effect, which enables us to show the remainder becomes more regular. Noticing
the singularity will disappear after initial time, the regularization estimate
together with $L^{2}$ decay of the short wave yields the $L^{\infty}$ decay of
the short wave. Combing this with the long wave, we finish the pointwise
structure inside the wave cone.

To get the global structure of the solution, we need the estimate outside the
wave cone, i.e., inside the space-like region. The weighted energy estimates
play a decisive role here. The weight functions are carefully chosen for
different $\gamma$'s. It is noted that the sufficient understanding of the
structure of the wave part obtained previously, is essential in the estimate.
Moreover, the regularization effect makes it possible to do the higher order
weighted energy estimate. Then the desired pointwise estimate follows from the
Sobolev inequality.

The rest of this paper is organized as follows: We first prepare some
important properties in Section \ref{pre} for the long wave-short wave
decomposition, the wave-remainder decomposition and regularization estimates.
Then we study the large time behavior in Section \ref{large-time}. Finally, we
study the initial layer and the asymptotic behavior in Section \ref{layer}.

\section{Preliminary}

\label{pre}

\subsection{The operator $L$}

First, we introduce a new norm $\left\vert \cdot\right\vert _{L_{\widetilde
{\sigma},\theta}^{2}}$:
\[
\left\vert g\right\vert _{L_{\widetilde{\sigma},\theta}^{2}}^{2}%
:=\int\left\langle v\right\rangle ^{2\theta}\left\vert \nabla g\right\vert
^{2}dv+\int\left\langle v\right\rangle ^{2\theta}\frac{\left\vert v\right\vert
^{2}\left\langle v\right\rangle ^{2\gamma-4}}{2}\left\vert g\right\vert
^{2}dv,\ \ \ \theta\in\mathbb{R},\ \gamma>0,
\]
which is equivalent to the natural norm $\left\vert \cdot\right\vert
_{L_{\sigma,\theta}}$. Through this equivalent norm, we can derive the
coercivity of the operator $L$ for all $\gamma>0,$ as below. The proof is
analogous to the Landau case \cite{[Guo]}.

\begin{lemma}
[Coercivity]\label{co} Let $\theta\in\mathbb{R}$, $\gamma>0$. For any $m>1,$
there is $0<C\left(  m\right)  <\infty,$ such that%
\begin{align}
&  \quad\left\vert \left\langle \left\langle v\right\rangle ^{2\theta}%
\frac{\triangle_{v}\Phi}{2}g_{1},g_{2}\right\rangle _{v}\right\vert
\label{1}\\
&  \leq\frac{C}{m^{\gamma}}\left\vert g_{1}\right\vert _{L_{\widetilde{\sigma
},\theta}^{2}}\left\vert g_{2}\right\vert _{L_{\widetilde{\sigma},\theta}^{2}%
}+C\left(  m\right)  \left(  \int_{\left\vert v\right\vert \leq m}\left\vert
\left\langle v\right\rangle ^{\theta}g_{1}\right\vert ^{2}dv\right)
^{1/2}\left(  \int_{\left\vert v\right\vert \leq m}\left\vert \left\langle
v\right\rangle ^{\theta}g_{2}\right\vert ^{2}dv\right)  ^{1/2}.\nonumber
\end{align}
Moreover, there exists $\nu_{0}>0$ such that%
\begin{equation}
\left\langle -Lg,g\right\rangle _{v}\geq\nu_{0}\left\vert \mathrm{P}%
_{1}g\right\vert _{L_{\sigma}^{2}}^{2}. \label{coercivity}%
\end{equation}

\end{lemma}

Now, let us decompose the collision operator $L=-\Lambda+K$, where
\[
\Lambda=-L+\varpi\chi_{R}\left(  |v|\right)  \,,\quad K=\varpi\chi_{R}\left(
|v|\right)  \,,
\]
here $\varpi>0$ and $R>0$ are as large as desired.

Regarding the behavior of solutions to equation (\ref{in.1.c}) in the
space-like region, the following weight functions $\mu(x,v)$ will be taken
into account:
\[
\mu(x,v)=1\quad\hbox{or}\quad\exp\left(  \left\langle x\right\rangle
/D\right)  \quad\hbox{if}\quad\ga\geq3/2\,,
\]
for $D$ large, and
\[
\mu(x,v)=1\quad\hbox{or}\quad\exp\left(  \alpha c(x,v)\right)  \quad
\hbox{if}\quad0<\ga<3/2\,,
\]
where
\begin{align*}
c(x,v)  &  =5\Big(\de\left\langle x\right\rangle \Big)^{\frac{\ga}{3-\ga}%
}\left(  1-\chi\left(  \de\left\langle x\right\rangle \left\langle
v\right\rangle ^{\ga-3}\right)  \right) \\
&  \quad+\bigg[\left(  1-\chi\left(  \de\left\langle x\right\rangle
\left\langle v\right\rangle ^{\ga-3}\right)  \right)  \de\left\langle
x\right\rangle \left\langle v\right\rangle ^{2\ga-3}+3\left\langle
v\right\rangle ^{\ga}\bigg]\chi\left(  \de\left\langle x\right\rangle
\left\langle v\right\rangle ^{\ga-3}\right)  \,,
\end{align*}
the positive constants $\delta$ and $\alpha$ being determined later.

\begin{lemma}
\label{prop1} Assuming that $\gamma>0$, we have the following properties of
the operators $\Lambda$ and $K$.

\noindent\textrm{(i)} There exists $c>0$ such that
\[
\int\left(  \Lambda g\right)  g\mu dxdv\geq c\left\Vert g\right\Vert
_{L_{\sigma}^{2}\left(  \mu\right)  }^{2}.
\]
\noindent\textrm{(ii)}
\[
\int\left(  Kg\right)  g\mu dxdv\leq\varpi\left\Vert g\right\Vert
_{L^{2}\left(  \mu\right)  }^{2}.
\]

\end{lemma}

\begin{proof}
We only prove part (i) when $\mu\left(  x,v\right)  =e^{\alpha c\left(
x,v\right)  },$ since the other cases of part (i) and part (ii) are trivial.
Notice that there is a constant $c_{1}>0$ such that%
\[
\left\Vert \nabla_{v}g\right\Vert _{L^{2}\left(  \mu\right)  }^{2}+\int\left[
\frac{\left\vert v\right\vert ^{2}\left\langle v\right\rangle ^{2\gamma-4}}%
{4}-\left(  \frac{3}{2}\left\langle v\right\rangle ^{\gamma-2}+\frac{\left(
\gamma-2\right)  }{2}\left\vert v\right\vert ^{2}\left\langle v\right\rangle
^{\gamma-4}\right)  +\varpi\chi_{R}\right]  g^{2}\mu dxdv\geq c_{1}\left\Vert
g\right\Vert _{L_{\sigma}^{2}\left(  \mu\right)  }^{2}%
\]
whenever $\varpi,$ $R>0$ are sufficiently large. On the other hand, it follows
from
\begin{equation}
\left\vert \nabla_{v}c\left(  x,v\right)  \right\vert \leq C\left(
\gamma\right)  \left\langle v\right\rangle ^{\gamma-1}\left(  1+\left\vert
\chi^{\prime}\left(  \left[  \delta\left\langle x\right\rangle \right]
\left\langle v\right\rangle ^{\gamma-3}\right)  \right\vert \right)  ,
\label{v-derivtive-c}%
\end{equation}
that
\[
\left\vert \int\nabla_{v}g\cdot\nabla_{v}\left(  \mu\right)  gdxdv\right\vert
\leq\alpha C\left(  \gamma\right)  \sup\left(  1+\left\vert \chi^{\prime
}\right\vert \right)  \int\left\langle v\right\rangle ^{\gamma-1}\left\vert
g\right\vert \left\vert \nabla_{v}g\right\vert \mu dxdv\leq\frac{\alpha Q}%
{2}\left\Vert g\right\Vert _{L_{\sigma}^{2}\left(  \mu\right)  }^{2},
\]
where $Q=\left[  C\left(  \gamma\right)  \sup\left(  1+\left\vert \chi
^{\prime}\right\vert \right)  \right]  .$ Therefore, we choose $\alpha>0$
sufficiently small with $\alpha Q<c_{1}$ and thus deduce that
\begin{align*}
\int\left(  \Lambda g\right)  g\mu dxdv  &  =\left\Vert \nabla_{v}g\right\Vert
_{L^{2}\left(  \mu\right)  }^{2}+\int\nabla_{v}g\cdot\nabla_{v}\left(
\mu\right)  gdxdv\\
&  \quad+\int\left[  \frac{\left\vert v\right\vert ^{2}\left\langle
v\right\rangle ^{2\gamma-4}}{4}-\left(  \frac{3}{2}\left\langle v\right\rangle
^{\gamma-2}+\frac{\left(  \gamma-2\right)  }{2}\left\vert v\right\vert
^{2}\left\langle v\right\rangle ^{\gamma-4}\right)  +\varpi\chi_{R}\right]
g^{2}\mu dxdv\\
&  \geq\frac{c_{1}}{2}\left\Vert g\right\Vert _{L_{\sigma}^{2}\left(
\mu\right)  }^{2}=c\left\Vert g\right\Vert _{L_{\sigma}^{2}\left(  \mu\right)
}^{2},
\end{align*}
which completes the proof of the lemma.
\end{proof}

On the other hand, in preparation for studying the time-like region, we
provide the spectrum $\mathrm{{Spec}(\eta)}$, $\eta\in{\mathbb{R}}^{3}$, of
the operator $L_{\eta}=-iv\cdot\eta+L$. In fact, we extend the results for the
case $\gamma=2$ in \cite{[LuoYu]} to the case $\gamma\geq1$.

\begin{lemma}
[Spectrum of $L_{\eta}$]\label{pr12} Assuming that $\gamma\geq1$, given
$0<\de\ll1,$

\noindent\textrm{(i)} There exists $\tau=\tau(\de)>0$ such that if
$|\eta|>\de$,
\begin{equation}
\hbox{\rm Spec}(\eta)\subset\{z\in\mathbb{C}:\mathrm{Re}(z)<-\tau\}\,.
\label{pre.ab.e}%
\end{equation}
\newline\noindent\textrm{(ii)} If $|\eta|<\de$,
\begin{equation}
\hbox{\rm Spec}(\eta)\cap\{z\in\mathbb{C}:\mathrm{Re}(z)>-\tau\}=\{\la(\eta
)\}\,, \label{pre.ab.f}%
\end{equation}
where $\lambda(\eta)$ is the eigenvalue of $L_{\eta}$ which is real and smooth
in $\eta$ only through $|\eta|^{2}$, i.e., $\lambda(\eta)=\mathscr{A}(|\eta
|^{2})$ for some real smooth function $\mathscr{A}$; the eigenfunction
$e_{D}(v,\eta)$ is smooth in $\eta$ as well. In addition, they are analytic in
$\eta$ if $\ga\geq3/2$. Their asymptotic expansions are given as below:
\begin{equation}%
\begin{array}
[c]{l}%
\label{pre.ab.h}\displaystyle\la(\eta)=-a_{\ga}|\eta|^{2}+O(|\eta|^{4})\,,\\
\\
\displaystyle e_{D}(\eta)=E_{D}+iE_{D,1}|\eta|+O(|\eta|^{2})\,,
\end{array}
\end{equation}
with $a_{\ga}>0$, $E_{D,1}=L^{-1}(v\cdot\om E_{D})$, $\om=\eta/|\eta|$. Here
$\{e_{D}(\eta)\}$ can be normalized by
\[
\big<e_{D}(-\eta),e_{D}(\eta)\big>_{v}=1\,.
\]
\newline\noindent\textrm{(iii)} Moreover, the semigroup $e^{(-i\eta\cdot
v+L)t}$ can be decomposed as
\begin{equation}
\displaystyle e^{(-i\eta\cdot v+L)t}f=e^{(-i\eta\cdot v+L)t}\Pi_{\eta}%
^{D\perp}f+\mathbf{1}_{\{|\eta|<\de\}}e^{\la(\eta)t}\big<e_{D}(-\eta
),f\big>_{v}e_{D}(\eta)\,, \label{pre.ab.g}%
\end{equation}
where $\mathbf{1}_{D}$ is the characteristic function of the domain $D$, and
there exist $a(\tau)>0$ and $\overline{a}>0$ such that $|e^{(-i\eta\cdot
v+L)t}\Pi_{\eta}^{D\perp}|_{L_{v}^{2}}\lesssim e^{-a(\tau)t}$ and
$|e^{\la(\eta)t}|\leq e^{-\overline{a}|\eta|^{2}t}$. \medskip
\end{lemma}

\begin{proof}
Let $L=-\Lambda+K$ with%
\[
\Lambda_{\eta}f=\left(  -\Lambda-i\eta\cdot v\right)  f,\ \ \ \ \ L_{\eta
}=\left(  L-i\eta\cdot v\right)  f,
\]
Here $f\in D\left(  \Lambda_{\eta}\right)  =\left\{  f\in L_{v}^{2}%
;\Lambda_{\eta}f\in L_{v}^{2}\right\}  $ and $D\left(  \Lambda_{\eta}\right)
=D\left(  L_{\eta}\right)  .$ Since $K$ is a bounded operator in $L_{v}^{2},$
$L_{\eta}\ $is regarded as a bounded perturbation of $\Lambda_{\eta}.$ We
shall verify that such a decomposition satisfies the four hypotheses
\textbf{H1-H4} stated in \cite{[YangYu]}. Under the assumptions \textbf{H1-H4,
}using semigroup theory and linear operator perturbation theory,\textbf{
}Theorem 1.1 in\textbf{ \cite{[YangYu]} }asserts that\textbf{ }the spectrum of
$L_{\eta}$\textbf{ }has the similar structure of the Boltzmann\textbf{
}equation with cutoff hard potential. Since the null space of the linear
Fokker-Planck operator is one-dimensional, for $\left\vert \eta\right\vert $
small enough, we only obtain one smooth eigenvalue of $L_{\eta}$ while there
are five smooth eigenvalues for the Boltzmann equation with cutoff hard
potentials. As to the verification of \textbf{H1-H4, }the proof is a slight
modification of the paper \cite{[LuoYu]} and hence we omit the details. The
hypothesis \textbf{H1} is worthy of being mentioned, for $\varpi$ sufficiently
large, there exists a constant $c>0$ such that
\[
\left\langle \Lambda f,f\right\rangle _{v}\geq c\left\vert f\right\vert
_{L_{\sigma}^{2}}^{2}\geq c\left\vert f\right\vert _{L_{v}^{2}}^{2},
\]
for all $\gamma\geq1,$ the last inequality holds since $\left\vert
f\right\vert _{L_{\sigma}^{2}}$ is stronger than $\left\vert f\right\vert
_{L_{v}^{2}}\ $as $\gamma\geq1.\ $This is why we miss the spectrum structure
for the case $0<\gamma<1.$


To prove (ii), we need to explore the symmetric properties of $\lambda(\eta)$
and $e_{D}(\eta)$. Here we follow the framework of section 7.3 in
\cite{[LiuYu1]}. First we notice there is a natural three dimensional
orthogonal group $O(3)$-action on $L_{v}^{2}$: Let $a\in O(3)$, $f\in
L_{v}^{2}$,
\[
(a\circ f)(v)\equiv f(a^{-1}v).
\]
Then it is easy to check the $O(3)$-action commutes with operators $L$,
$\mathrm{P}_{0}$ and $\mathrm{P}_{1}$. Consider the eigenvalue problem
\begin{equation}
L_{\eta}e_{D}(\eta)=(-iv\cdot\eta+L)e_{D}(\eta)=\lambda(\eta)e_{D}(\eta).
\label{eigenprob}%
\end{equation}
Apply $a\in O(3)$ to \eqref{eigenprob}, by commutative properties and the fact
that $a$ preserves the vector inner product in $\mathbb{R}^{3}$,
\[
(-iv\cdot(a\eta)+L)(a\circ e_{D}(\eta))=\lambda(\eta)(a\circ e_{D})(\eta).
\]
Then $\lambda(a\eta)=\lambda(\eta)$, $e_{D}(a\eta)=a\circ e(\eta)$, which
implies that $\lambda(\eta)$ is dependent only upon $|\eta|$. Now let $a\in
O(3)$ be an orthogonal transformation that sends $\frac{\eta}{|\eta|}$ to
$(1,0,0)^{T}$. Thus the original eigenvalue problem \eqref{eigenprob} is
reduced to
\begin{equation}
(-iv_{1}|\eta|+L)e(|\eta|)=\lambda(|\eta|)e(|\eta|), \label{eigenprob1}%
\end{equation}
with $\lambda(\eta)=\lambda(|\eta|)$, $e_{D}(\eta)=a^{-1}\circ e_{D}(|\eta
|)$.
Apply the Macro-Micro decomposition to \eqref{eigenprob1} to yield
\begin{subequations}
\label{eq:L6.1}%
\begin{align}
&  -i|\eta|\mathrm{P}_{0}v_{1}\big(\mathrm{P}_{0}e+\mathrm{P}_{1}%
e\big)=\lambda\mathrm{P}_{0}e,\label{eq:L6.1a}\\
&  -i|\eta|\mathrm{P}_{1}v_{1}\mathrm{P}_{0}e-i|\eta|\mathrm{P}_{1}%
v_{1}\mathrm{P}_{1}e+L\mathrm{P}_{1}e=\lambda\mathrm{P}_{1}e. \label{eq:L6.1b}%
\end{align}
Set $\lambda(|\eta|)=i|\eta|\zeta(|\eta|).$ We can solve $\mathrm{P}_{1}e$ in
terms of $\mathrm{P}_{0}e$ from \eqref{eq:L6.1b},
\end{subequations}
\begin{equation}
\mathrm{P}_{1}e=i|\eta|\big[L-i|\eta|\mathrm{P}_{1}v_{1}-i|\eta|\zeta
(|\eta|)\big]^{-1}\mathrm{P}_{1}v_{1}\mathrm{P}_{0}e, \label{eq:L6.2}%
\end{equation}
then substitute this back to \eqref{eq:L6.1a} to get
\begin{equation}
\Big(\mathrm{P}_{0}v_{1}+i|\eta|\mathrm{P}_{0}\big[L-i|\eta|\mathrm{P}%
_{1}v_{1}-i|\eta|\zeta(|\eta|)\big]^{-1}\mathrm{P}_{1}v_{1}\Big)\mathrm{P}%
_{0}e=-\zeta\mathrm{P}_{0}e. \label{eigenprob_fte}%
\end{equation}
We notice that this is actually a finite dimensional eigenvalue problem. The
solvability of it and the asymptotic expansions of eigenvalue and
eigenfunction for $|\eta|\ll1$ are essentially due to the implicit function
theorem. The procedure is basically the same as the case $\gamma=2,$ we refer
the readers to Theorem 3.2 in \cite{[LuoYu]} for details. We obtain
$\lambda(|\eta|)$ and $\mathrm{P}_{0}e(|\eta|)=\beta(|\eta|)E_{D}$ with
$\lambda$ and $\beta$ being smooth functions. Furthermore, $\lambda(|\eta|)$
and $\beta(|\eta|)$ are not merely smooth but analytic for $\gamma\geq3/2$. To
prove this, it suffices to check that the perturbation $ivf$ is $L$-bounded,
i.e.,
\[
|vf|_{L_{v}^{2}}^{2}\leq C_{1}|Lf|_{L_{v}^{2}}^{2}+C_{2}|f|_{L_{v}^{2}}%
^{2}\,.
\]
Then the Kato-Rellich theorem guarantees the operator $B(z)=-iv_{1}z+L$ is in
the analytic family of Type (A), see \cite{[Kato]}, which in turn implies the
eigenvalue and eigenfunction associated with \eqref{eigenprob1} are analytic
in $|\eta|$, cf. \cite{[DeLe]}. Now, let us calculate $\big<\Lambda f,\Lambda
f\big>_{v}$ first. For simplicity of notation, let
\[
\psi(v)=\frac{1}{4}|v|^{2}\left\langle v\right\rangle ^{2\ga-4}-\left(
\frac{3}{2}\left\langle v\right\rangle ^{\ga-2}+\frac{\ga-2}{2}|v|^{2}%
\left\langle v\right\rangle ^{\ga-4}\right)  +\varpi\chi_{R}(|v|)\,,
\]
then
\begin{align*}
\big<\Lambda f,\Lambda f\big>_{v}  &  =|\Delta_{v}f|_{L_{v}^{2}}^{2}%
+|\psi(v)f|_{L_{v}^{2}}^{2}\\
&  \quad+2\big<\psi(v),\big(\nabla_{v}f,\nabla_{v}f\big)\big>_{v}%
+2\big<f,\big(\nabla_{v}\psi(v),\nabla_{v}f\big)\big>_{v}\,.
\end{align*}
By the Cauchy inequality, we have
\[
\big|\big<f,\big(\nabla_{v}\psi(v),\nabla_{v}f\big)\big>_{v}\big|\leq
\big<\psi(v),\big(\nabla_{v}f,\nabla_{v}f\big)\big>_{v}+\frac{1}{4}%
\Big<\frac{f^{2}}{\psi(v)},\big(\nabla_{v}\psi(v),\nabla_{v}\psi
(v)\big)\Big>_{v}\,.
\]
Let us compare $\big(\nabla_{v}\psi(v),\nabla_{v}\psi(v)\big)$ and $\psi
^{3}(v)$. For $|v|$ large, we have
\[
\big(\nabla_{v}\psi(v),\nabla_{v}\psi(v)\big)\approx|v|^{4\ga-6}%
\]
and
\[
\psi^{3}(v)\approx|v|^{6\ga-6}\,.
\]
For $|v|$ small, one can choose $\varpi$ large enough such that
\[
\big(\nabla_{v}\psi(v),\nabla_{v}\psi(v)\big)\ll\psi^{3}(v)\,.
\]
This means
\[
\big<\Lambda f,\Lambda f\big>_{v}\geq|\Delta_{v}f|_{L_{v}^{2}}^{2}+\frac{1}%
{2}|\psi(v)f|_{L_{v}^{2}}^{2}\gtrsim|\left\langle v\right\rangle
^{2\ga-2}f|_{L_{v}^{2}}^{2}\,.
\]
Hence if $\ga\geq3/2$,
\begin{align*}
|vf|_{L_{v}^{2}}^{2}\leq C\big<\Lambda f,\Lambda f\big>_{v}  &
=C\big<Lf-Kf,Lf-Kf\big>_{v}\\
&  \leq C_{1}|Lf|_{L_{v}^{2}}^{2}+C_{2}|f|_{L_{v}^{2}}^{2}\,.
\end{align*}
However, we cannot simply deduce smoothness (analyticity) in $\eta$ from
smoothness (analyticity) in $|\eta|$. Our goal is to show $\lambda(z)$ and
$\beta(z)$ are in fact even functions in $z$. If so, due to a classical
theorem of Whitney \cite{Whitney}, we have
\[
\lambda(|\eta|)=\mathscr{A}(|\eta|^{2}),\qquad\beta(|\eta|)=\mathscr{B}(|\eta
|^{2}),
\]
for some smooth or analytic functions $\mathscr{A}$ and $\mathscr{B}$ provided
$\lambda(|\eta|)$ and $\beta(|\eta|)$ are smooth or analytic respectively. To
show they are even, let us define a map $\mathcal{R}:(v_{1},v_{2}%
,v_{3})\mapsto(-v_{1},v_{2},v_{3})$, then obviously $\mathcal{R}\in O(3)$. We
apply $\mathcal{R}$ to \eqref{eigenprob1},
\[
(-iv_{1}(-|\eta|)+L)(\mathcal{R}\circ e(|\eta|))=\lambda(|\eta|)(\mathcal{R}%
\circ e(|\eta|)),
\]
which is an eigenvalue problem with $|\eta|\rightarrow-|\eta|$. This follows
that the eigenpair $\{\lambda(|\eta|),\mathcal{R}\circ e(|\eta|)\}$ coincides
with $\{(\lambda(-|\eta|),e(-|\eta|))\}$. Hence
\begin{equation}
\lambda(|\eta|)=\lambda(-|\eta|),\qquad\mathcal{R}\circ e(|\eta|)=e(-|\eta|).
\label{eq:L6.3}%
\end{equation}
In addition, use $\mathcal{R}\circ\mathrm{P}_{0}e(|\eta|)=\mathrm{P}%
_{0}\mathcal{R}\circ e(|\eta|)=\mathrm{P}_{0}e(-|\eta|)$ and $\mathcal{R}\circ
E_{D}=E_{D}$ to find $\beta(|\eta|)=\beta(-|\eta|)$, namely $\beta$ is also an
even function. We can show
\begin{equation}
\overline{\lambda(|\eta|)}=\lambda(-|\eta|),\qquad\overline{e(|\eta
|)}=e(-|\eta|), \label{eq:L6.4}%
\end{equation}
by taking the complex conjugate of \eqref{eigenprob1}. This together with
\eqref{eq:L6.3} shows $\lambda(|\eta|)$ and $\beta(|\eta|)$ are real
functions. By \eqref{eq:L6.2}, we can construct $e(|\eta|)$ from
$\mathrm{P}_{0}e(|\eta|)$,
\[
e(|\eta|)=\mathrm{P}_{0}e(|\eta|)+\mathrm{P}_{1}e(|\eta|)=\Big(1+i|\eta
|\big[L-i|\eta|\mathrm{P}_{1}v_{1}-\lambda(|\eta|)\big]^{-1}\mathrm{P}%
_{1}v_{1}\Big)\beta(|\eta|)E_{D}.
\]
The eigenfunction $e_{D}(\eta)$ to the original eigenvalue-problem
\eqref{eigenprob} can be recovered by applying $a^{-1}$, namely
\[
e_{D}(\eta)=a^{-1}\circ e(|\eta|)=\Big(1+\big[L-\mathrm{P}_{1}i\eta\cdot
v-\mathscr{A}(|\eta|^{2})\big]^{-1}\mathrm{P}_{1}i\eta\cdot
v\Big)\mathscr{B}(|\eta|^{2})E_{D}.
\]
Therefore the proof is complete.
\end{proof}

\subsection{The semigroup operator $e^{t\mathcal{L}}$}

Now, let $h$ be the solution of the equation
\begin{equation}
\left\{
\begin{array}
[c]{l}%
\pa_{t}h=\mathcal{L}h,\quad\hbox{where}\quad\mathcal{L}h=-v\cdot\nabla
_{x}h-\Lambda h\,,\\[4mm]%
h(0,x,v)=h_{0}(x,v)\,.
\end{array}
\right.
\end{equation}
In this subsection, we will study some properties of the semigroup operator
$e^{t\mathcal{L}}$.

\begin{lemma}
\label{x-energy}For any $k\in{\mathbb{N}}\cup\{0\},$

\noindent\textrm{(i)} If $\gamma\geq1$, there exists $C>0$ such that
\begin{equation}
\left\Vert e^{t\mathcal{L}}h_{0}\right\Vert _{H_{x}^{k}L_{v}^{2}(\mu)}\lesssim
e^{-Ct}\left\Vert h_{0}\right\Vert _{H_{x}^{k}L_{v}^{2}(\mu)}\,.
\label{energy1}%
\end{equation}
\noindent\textrm{(ii)} If $0<\gamma<1$, we have
\begin{equation}
\left\Vert e^{t\mathcal{L}}h_{0}\right\Vert _{H_{x}^{k}L_{v}^{2}(\mu)}%
\lesssim\left\Vert h_{0}\right\Vert _{H_{x}^{k}L_{v}^{2}(\mu)}\,.
\label{energy2}%
\end{equation}

\end{lemma}

\begin{proof}
It suffices to show that there exists $c_{0}>0$ such that for any multi-index
$\beta,$%
\begin{equation}
\frac{d}{dt}\left\Vert \partial_{x}^{\beta}h\right\Vert _{L^{2}\left(
\mu\right)  }^{2}\leq-c_{0}\left\Vert \partial_{x}^{\beta}h\right\Vert
_{L_{\sigma}^{2}\left(  \mu\right)  }^{2}. \label{Energy-evolu}%
\end{equation}
In view of Lemma \ref{prop1}, we have
\begin{align*}
\frac{1}{2}\frac{d}{dt}\left\Vert \partial_{x}^{\beta}h\right\Vert
_{L^{2}\left(  \mu\right)  }^{2}  &  =\int-v\cdot\nabla_{x}\left(
\partial_{x}^{\beta}h\right)  \partial_{x}^{\beta}h\mu dxdv-\int\Lambda\left(
\partial_{x}^{\beta}h\right)  \partial_{x}^{\beta}h\mu dxdv\\
&  \leq\int\frac{1}{2}\left(  \partial_{x}^{\beta}h\right)  ^{2}v\cdot
\nabla_{x}\mu dxdv-c_{0}\left\Vert \partial_{x}^{\beta}h\right\Vert
_{L_{\sigma}^{2}\left(  \mu\right)  }^{2}.
\end{align*}

If $\mu\left(  x,v\right)  \equiv1,$ $\left(  \ref{Energy-evolu}\right)  $ is
obvious since $\int\frac{1}{2}\left(  \partial_{x}^{\beta}h\right)  ^{2}%
v\cdot\nabla_{x}\mu dxdv=0.$

If $\mu\left(  x,v\right)  =\exp\left(  \left\langle x\right\rangle /D\right)
$\ and $\gamma\geq3/2$, we choose $D$ sufficiently large such that
$1/D<\min\{c_{0},1\}\ $and thus obtain
\begin{align*}
\left\vert \int\frac{1}{2}\left(  \partial_{x}^{\beta}h\right)  ^{2}%
v\cdot\nabla_{x}\mu dxdv\right\vert  &  =\left\vert \int\frac{1}{2}\left(
\partial_{x}^{\beta}h\right)  ^{2}v\cdot\frac{x}{D\left\langle x\right\rangle
}\mu dxdv\right\vert \\
&  \leq\frac{1}{2D}\int\left\langle v\right\rangle ^{2\gamma-2}\left(
\partial_{x}^{\beta}h\right)  ^{2}\mu dxdv\leq\frac{c_{0}}{2}\left\Vert
\partial_{x}^{\beta}h\right\Vert _{L_{\sigma}^{2}\left(  \mu\right)  }^{2}.
\end{align*}

If $\mu\left(  x,v\right)  =e^{\alpha c\left(  x,v\right)  }$ and $0<$
$\gamma<3/2,$ since
\[
\left\vert \nabla_{x}c\left(  x,v\right)  \right\vert \leq\delta C\left\langle
v\right\rangle ^{2\gamma-3},
\]
for some constant $C>0$ depending only upon $\gamma,$ we have
\[
\left\vert \int\frac{1}{2}\left(  \partial_{x}^{\beta}h\right)  ^{2}%
v\cdot\nabla_{x}\mu dxdv\right\vert \leq\frac{\delta C\alpha}{2}%
\int\left\langle v\right\rangle ^{2\gamma-2}\left(  \partial_{x}^{\beta
}h\right)  ^{2}\mu dxdv\leq\frac{c_{0}}{2}\left\Vert \partial_{x}^{\beta
}h\right\Vert _{L_{\sigma}^{2}\left(  \mu\right)  }^{2}%
\]
by choosing $\alpha,\delta$ $>0\ $small enough with $0<$ $\delta C\alpha
<\min\{c_{0},1\}.$

Grouping the above discussions, we obtain $\left(  \ref{Energy-evolu}\right)
$ and thus deduce that for $\gamma>0,$ $k\in{\mathbb{N\cup\{}}0{\mathbb{\}}}%
$,
\[
\left\Vert e^{t\mathcal{L}}h_{0}\right\Vert _{H_{x}^{k}L_{v}^{2}(\mu)}%
\leq\left\Vert h_{0}\right\Vert _{H_{x}^{k}L_{v}^{2}(\mu)}.
\]
Moreover, if $\gamma\geq1,$ then $\left(  \ref{Energy-evolu}\right)  $
becomes
\[
\frac{1}{2}\frac{d}{dt}\left\Vert \partial_{x}^{\beta}h\right\Vert
_{L^{2}\left(  \mu\right)  }^{2}\leq-\frac{c_{0}}{2}\left\Vert \partial
_{x}^{\beta}h\right\Vert _{L_{\sigma}^{2}\left(  \mu\right)  }^{2}\leq
-\frac{c_{0}}{2}\left\Vert \partial_{x}^{\beta}h\right\Vert _{L^{2}\left(
\mu\right)  }^{2},
\]
which leads to the exponential time decay of all $x$-derivatives of the
solution $e^{t\mathcal{L}}h_{0}$ in the weighted $L^{2}$ norm.
\end{proof}

The following is the regularization estimate of the semigroup operator
$e^{t\mathcal{L}}$ in small time.

\begin{lemma}
[Regularization estimate]\label{regularization} For $\gamma>0$ and $0<t\leq1$,
we have
\[
\int|\nabla_{v}e^{t\mathcal{L}}h_{0}|^{2}\mu dxdv=O(t^{-1})\int|h_{0}|^{2}\mu
dxdv
\]
and
\[
\int|\nabla_{x}e^{t\mathcal{L}}h_{0}|^{2}\mu dxdv=O(t^{-3})\int|h_{0}|^{2}\mu
dxdv\,.
\]

\end{lemma}

\begin{proof}
Recall that $\Lambda=-L+K$ with%
\[
\Lambda f=-\triangle_{v}f+\left[  \frac{1}{4}\left\vert \nabla_{v}%
\Phi\right\vert ^{2}-\frac{1}{2}\triangle_{v}\Phi\right]  f+\varpi\chi
_{R}\left(  \left\vert v\right\vert \right)  f,\ \ \ K=\varpi\chi_{R}\left(
\left\vert v\right\vert \right)  f.
\]
Here we choose $R>0$ and $\varpi>0$ sufficiently large such that%
\begin{equation}
\ \left\{
\begin{array}
[c]{l}%
\frac{1}{4}\left\vert \nabla_{v}\Phi\right\vert ^{2}-\frac{1}{2}\triangle
_{v}\Phi\leq\frac{1}{3}\left\langle v\right\rangle ^{2\gamma-2}%
,\ \ \ \left\vert \nabla_{v}\left(  \frac{1}{4}\left\vert \nabla_{v}%
\Phi\right\vert ^{2}-\frac{1}{2}\triangle_{v}\Phi\right)  \right\vert
\lesssim\left\langle v\right\rangle ^{2\gamma-3}\ \ \text{for }\left\vert
v\right\vert >R,\vspace{3mm}\\
\left\vert \frac{1}{4}\left\vert \nabla_{v}\Phi\right\vert ^{2}-\frac{1}%
{2}\triangle_{v}\Phi\right\vert ,\ \left\vert \nabla_{v}\left(  \frac{1}%
{4}\left\vert \nabla_{v}\Phi\right\vert ^{2}-\frac{1}{2}\triangle_{v}%
\Phi\right)  \right\vert <\frac{\varpi}{2}\ \ \text{for }\left\vert
v\right\vert \leq R.
\end{array}
\right.  \label{*}%
\end{equation}
Now, define the energy functional
\[
\mathcal{F}\left(  t,h_{t}\right)  :=A\left\Vert h\right\Vert _{L^{2}\left(
\mu\right)  }^{2}+at\left\Vert \nabla_{v}h\right\Vert _{L^{2}\left(
\mu\right)  }^{2}+2ct^{2}\left\langle \nabla_{x}h,\nabla_{v}h\right\rangle
_{L^{2}\left(  \mu\right)  }+bt^{3}\left\Vert \nabla_{x}h\right\Vert
_{L^{2}\left(  \mu\right)  }^{2},
\]
with $a,$ $b,$ $c>0$ and $c<\sqrt{ab}$ (positive definite) and $A>0$
sufficiently large. We shall show that $d\mathcal{F}/dt\leq0,\ t\in\left(
0,1\right)  ,$ via choosing suitable positive constants $A,$ $a,$ $b$ and $c.$

In $\left(  \ref{Energy-evolu}\right)  ,$ it has been shown that
\begin{equation}
\frac{d}{dt}\left\Vert h\right\Vert _{L^{2}(\mu)}^{2}\leq-c_{0}\left\Vert
h\right\Vert _{L_{\sigma}^{2}(\mu)}^{2},
\end{equation}%
\begin{equation}
\frac{d}{dt}\left\Vert \partial_{x_{i}}h\right\Vert _{L^{2}(\mu)}^{2}%
\leq-c_{0}\left\Vert \partial_{x_{i}}h\right\Vert _{L_{\sigma}^{2}(\mu)}^{2}.
\end{equation}
Next, we show that
\begin{equation}
\frac{1}{2}\frac{d}{dt}\left\Vert \partial_{v_{i}}h\right\Vert _{L^{2}\left(
\mu\right)  }^{2}\leq-\int\partial_{x_{i}}h\partial_{v_{i}}h\mu dxdv-\frac
{c_{0}}{2}\left\Vert \partial_{v_{i}}h\right\Vert _{L_{\sigma}^{2}\left(
\mu\right)  }^{2}+C_{\varepsilon}\left\Vert h\right\Vert _{L_{\sigma}%
^{2}\left(  \mu\right)  }^{2}+\varepsilon\left\Vert \partial_{v_{i}%
}h\right\Vert _{L_{\sigma}^{2}\left(  \mu\right)  }^{2},
\end{equation}
where $\varepsilon>0$ is arbitrarily small and $C_{\varepsilon}=O\left(
1/\varepsilon\right)  .$ Compute%
\begin{align*}
\frac{1}{2}\frac{d}{dt}\left\Vert \partial_{v_{i}}h\right\Vert _{L^{2}\left(
\mu\right)  }^{2}  &  =-\int\partial_{x_{i}}h\partial_{v_{i}}h\mu
dxdv+\int\left[  \frac{v}{2}\left(  \partial_{v_{i}}h\right)  ^{2}\right]
\cdot\nabla_{x}\mu dxdv\\
&  \quad-\int\left(  \Lambda\partial_{v_{i}}h\right)  \partial_{v_{i}}h\mu
dxdv-\int\left(  \left[  \partial_{v_{i}},\Lambda\right]  h\right)
\partial_{v_{i}}h\mu dxdv,
\end{align*}
where%
\[
\left[  \partial_{v_{i}},\Lambda\right]  h=\left[  \partial_{v_{i}}\left(
\frac{1}{4}\left\vert \nabla_{v}\Phi\right\vert ^{2}-\frac{1}{2}\triangle
_{v}\Phi\right)  \right]  h+\varpi\partial_{v_{i}}\left(  \chi_{R}\right)  h.
\]
In the course of the proof of Lemma \ref{x-energy}, one can see that
\begin{equation}
\left\vert \nabla_{x}\mu\right\vert \leq\min\{c_{0},1\}\cdot\left\langle
v\right\rangle ^{2\gamma-3}\mu\ \ \ \text{and}\ \ \ \ \ \ \left\vert
\nabla_{x}\mu\right\vert \leq\min\{c_{0},1\}\mu, \label{x-derivative-mu}%
\end{equation}
so
\[
\left\vert \int\left[  \frac{v}{2}\left(  \partial_{v_{i}}h\right)
^{2}\right]  \cdot\nabla_{x}\mu dxdv\right\vert \leq\frac{c_{0}}{2}%
\int\left\langle v\right\rangle ^{2\gamma-2}\left(  \partial_{v_{i}}h\right)
^{2}\mu dxdv\leq\frac{c_{0}}{2}\left\Vert \partial_{v_{i}}h\right\Vert
_{L_{\sigma}^{2}\left(  \mu\right)  }^{2}.
\]
Furthermore, by $\left(  \ref{*}\right)  $ we obtain
\begin{align*}
\left\vert \int\left(  \left[  \partial_{v_{i}},\Lambda\right]  h\right)
\partial_{v_{i}}h\mu dxdv\right\vert  &  \leq C\int\left\langle v\right\rangle
^{2\gamma-2}\left\vert h\partial_{v_{i}}h\right\vert \mu dxdv+\frac{\varpi}%
{2}\int\chi_{R}\left\vert h\partial_{v_{i}}h\right\vert \mu dxdv+\frac{\varpi
}{R}\int\left\vert \chi_{R}^{\prime}\right\vert \left\vert h\partial_{v_{i}%
}h\right\vert \mu dxdv\\
&  \leq C^{\prime}\int\left\langle v\right\rangle ^{2\gamma-2}\left\vert
h\partial_{v_{i}}h\right\vert \mu dxdv\leq C_{\varepsilon}\left\Vert
h\right\Vert _{L_{\sigma}^{2}\left(  \mu\right)  }^{2}+\varepsilon\left\Vert
\partial_{v_{i}}h\right\Vert _{L_{\sigma}^{2}\left(  \mu\right)  }^{2},
\end{align*}
where $\varepsilon>0$ is arbitrary small and $C_{\varepsilon}=O\left(
1/\varepsilon\right)  .$ It turns out that
\[
\frac{1}{2}\frac{d}{dt}\left\Vert \partial_{v_{i}}h\right\Vert _{L^{2}(\mu
)}^{2}\leq-\int\partial_{x_{i}}h\partial_{v_{i}}h\mu dxdv-\frac{c_{0}}%
{2}\left\Vert \partial_{v_{i}}h\right\Vert _{L_{\sigma}^{2}(\mu)}%
^{2}+C_{\varepsilon}\left\Vert h\right\Vert _{L_{\sigma}^{2}\left(
\mu\right)  }^{2}+\varepsilon\left\Vert \partial_{v_{i}}h\right\Vert
_{L_{\sigma}^{2}\left(  \mu\right)  }^{2}.
\]

Finally, direct computation gives
\begin{align*}
&  \quad\frac{d}{dt}\int\partial_{x_{i}}h\partial_{v_{i}}h\mu dxdv\\
&  =-\int\left\vert \partial_{x_{i}}h\right\vert ^{2}\mu dxdv-2\int\nabla
_{v}\left(  \partial_{x_{i}}h\right)  \cdot\nabla_{v}\left(  \partial_{v_{i}%
}h\right)  \mu dxdv-2\int\left(  \frac{\left\vert \nabla_{v}\Phi\right\vert
^{2}}{4}-\frac{\triangle_{v}\Phi}{2}+\varpi\chi_{R}\right)  \partial_{x_{i}%
}h\partial_{v_{i}}h\mu dxdv\\
&  \quad+\int\left(  v\cdot\nabla_{x}\mu\right)  \partial_{v_{i}}%
h\partial_{x_{i}}hdxdv+\frac{1}{2}\int\partial_{v_{i}}\left(  \frac{\left\vert
\nabla_{v}\Phi\right\vert ^{2}}{4}-\frac{\triangle_{v}\Phi}{2}+\varpi\chi
_{R}\right)  h^{2}\partial_{x_{i}}\mu dxdv\\
&  \quad-\int\nabla_{v}\left(  \partial_{x_{i}}h\right)  \cdot\nabla
_{v}\left(  \mu\right)  \partial_{v_{i}}hdxdv-\int\nabla_{v}\left(
\partial_{v_{i}}h\right)  \cdot\nabla_{v}\left(  \mu\right)  \partial_{x_{i}%
}hdxdv.
\end{align*}
From $\left(  \ref{v-derivtive-c}\right)  $, it follows that%
\begin{align*}
&  \quad\left\vert \int\nabla_{v}\left(  \partial_{x_{i}}h\right)  \cdot
\nabla_{v}\left(  \mu\right)  \partial_{v_{i}}hdxdv+\int\nabla_{v}\left(
\partial_{v_{i}}h\right)  \cdot\nabla_{v}\left(  \mu\right)  \partial_{x_{i}%
}hdxdv\right\vert \\
&  \leq\alpha C\left(  \gamma\right)  \sup\left(  1+\left\vert \chi^{\prime
}\right\vert \right)  \int\left\langle v\right\rangle ^{\gamma-1}\left(
\left\vert \nabla_{v}\left(  \partial_{x_{i}}h\right)  \right\vert \left\vert
\partial_{v_{i}}h\right\vert +\left\vert \nabla_{v}\left(  \partial_{v_{i}%
}h\right)  \right\vert \left\vert \partial_{x_{i}}h\right\vert \right)  \mu
dxdv\\
&  \leq\int\left(  \left\langle v\right\rangle ^{\gamma-1}\left\vert
\nabla_{v}\left(  \partial_{x_{i}}h\right)  \right\vert \left\vert
\partial_{v_{i}}h\right\vert +\left\langle v\right\rangle ^{\gamma
-1}\left\vert \nabla_{v}\left(  \partial_{v_{i}}h\right)  \right\vert
\left\vert \partial_{x_{i}}h\right\vert \right)  \mu dxdv,
\end{align*}
from $\alpha C\left(  \gamma\right)  \sup\left(  1+\left\vert \chi^{\prime
}\right\vert \right)  =\alpha Q<1.$ Note that this inequality is valid in the
cases $\mu\left(  x,v\right)  =1$ and $\mu\left(  x,v\right)  =\exp\left(
\left\langle x\right\rangle /D\right)  $ as well, since $\nabla_{v}\mu=0$ in
both cases. Therefore,%
\begin{align*}
&  \quad\frac{d}{dt}\int\partial_{x_{i}}h\partial_{v_{i}}h\mu dxdv\\
&  \leq-\int\left\vert \partial_{x_{i}}h\right\vert ^{2}\mu dxdv-2\int
\nabla_{v}\left(  \partial_{x_{i}}h\right)  \cdot\nabla_{v}\left(
\partial_{v_{i}}h\right)  \mu dxdv-2\int\left(  \frac{\left\vert \nabla
_{v}\Phi\right\vert ^{2}}{4}-\frac{\triangle_{v}\Phi}{2}+\varpi\chi
_{R}\right)  \partial_{x_{i}}h\partial_{v_{i}}h\mu dxdv\\
&  \quad+\int\left(  v\cdot\nabla_{x}\mu\right)  \partial_{v_{i}}%
h\partial_{x_{i}}hdxdv+\frac{1}{2}\int\partial_{v_{i}}\left(  \frac{\left\vert
\nabla_{v}\Phi\right\vert ^{2}}{4}-\frac{\triangle_{v}\Phi}{2}+\varpi\chi
_{R}\right)  h^{2}\partial_{x_{i}}\mu dxdv\\
&  \quad+\int\left(  \left\langle v\right\rangle ^{\gamma-1}\left\vert
\nabla_{v}\left(  \partial_{x_{i}}h\right)  \right\vert \left\vert
\partial_{v_{i}}h\right\vert +\left\langle v\right\rangle ^{\gamma
-1}\left\vert \nabla_{v}\left(  \partial_{v_{i}}h\right)  \right\vert
\left\vert \partial_{x_{i}}h\right\vert \right)  \mu dxdv.
\end{align*}

Collecting terms gives%
\begin{align*}
&  \quad\frac{d}{dt}\mathcal{F}\left(  t,h_{t}\right) \\
&  \leq-c_{0}A\left\Vert h\right\Vert _{L_{\sigma}^{2}\left(  \mu\right)
}^{2}+a\left\Vert \nabla_{v}h\right\Vert _{L^{2}\left(  \mu\right)  }%
^{2}+4ct\left\langle \nabla_{x}h,\nabla_{v}h\right\rangle _{L^{2}\left(
\mu\right)  }+3bt^{2}\left\Vert \nabla_{x}h\right\Vert _{L^{2}\left(
\mu\right)  }^{2}\\
&  \quad+2at\left[  -\sum_{i=1}^{3}\int\partial_{x_{i}}h\partial_{v_{i}}h\mu
dxdv-\frac{c_{0}}{2}\left\Vert \nabla_{v}h\right\Vert _{L_{\sigma}^{2}\left(
\mu\right)  }^{2}+3C_{\varepsilon}\left\Vert h\right\Vert _{L_{\sigma}%
^{2}\left(  \mu\right)  }^{2}+\varepsilon\left\Vert \nabla_{v}h\right\Vert
_{L_{\sigma}^{2}\left(  \mu\right)  }^{2}\right] \\
&  \quad+2ct^{2}\left[  -\int\left\vert \nabla_{x}h\right\vert ^{2}\mu
dxdv-2\sum_{i=1}^{3}\int\nabla_{v}\left(  \partial_{x_{i}}h\right)
\cdot\nabla_{v}\left(  \partial_{v_{i}}h\right)  \mu dxdv\right. \\
&  \quad-2\sum_{i=1}^{3}\int\left(  \frac{\left\vert \nabla_{v}\Phi\right\vert
^{2}}{4}-\frac{\triangle_{v}\Phi}{2}+\varpi\chi_{R}\right)  \partial_{x_{i}%
}h\partial_{v_{i}}h\mu dxdv-\sum_{i=1}^{3}\int\left(  v\cdot\nabla_{x}%
\mu\right)  \partial_{v_{i}}h\partial_{x_{i}}hdxdv\\
&  \quad+\frac{1}{2}\sum_{i=1}^{3}\int\partial_{v_{i}}\left(  \frac{\left\vert
\nabla_{v}\Phi\right\vert ^{2}}{4}-\frac{\triangle_{v}\Phi}{2}+\varpi\chi
_{R}\right)  h^{2}\partial_{x_{i}}\mu dxdv\\
&  \quad\left.  +\sum_{i=1}^{3}\int\left(  \left\langle v\right\rangle
^{\gamma-1}\left\vert \nabla_{v}\left(  \partial_{x_{i}}h\right)  \right\vert
\left\vert \partial_{v_{i}}h\right\vert +\left\langle v\right\rangle
^{\gamma-1}\left\vert \nabla_{v}\left(  \partial_{v_{i}}h\right)  \right\vert
\left\vert \partial_{x_{i}}h\right\vert \right)  \mu dxdv\right]  -bc_{0}%
t^{3}\left\Vert \nabla_{x}h\right\Vert _{L_{\sigma}^{2}\left(  \mu\right)
}^{2}.
\end{align*}
By $\left(  \ref{*}\right)  $ and $\left(  \ref{x-derivative-mu}\right)  $,
\begin{align*}
&  \quad\left\vert -2\sum_{i=1}^{3}\int\left(  \frac{\left\vert \nabla_{v}%
\Phi\right\vert ^{2}}{4}-\frac{\triangle_{v}\Phi}{2}+\varpi\chi_{R}\right)
\partial_{x_{i}}h\partial_{v_{i}}h\mu dxdv+\sum_{i=1}^{3}\int\left(
v\cdot\nabla_{x}\mu\right)  \partial_{v_{i}}h\partial_{x_{i}}hdxdv\right\vert
\\
&  \leq\sum_{i=1}^{3}2\int\left\langle v\right\rangle ^{2\gamma-2}\left\vert
\partial_{x_{i}}h\partial_{v_{i}}h\right\vert \mu dxdv+3\varpi\int\chi
_{R}\left\vert \nabla_{x}h\cdot\nabla_{v}h\right\vert \mu dxdv\\
&  \leq\sum_{i=1}^{3}\left(  \frac{bc_{0}}{8c}t\left\Vert \left\langle
v\right\rangle ^{\gamma-1}\left(  \partial_{x_{i}}h\right)  \right\Vert
_{L^{2}\left(  \mu\right)  }^{2}+\frac{8c}{bc_{0}t}\left\Vert \left\langle
v\right\rangle ^{\gamma-1}\left(  \partial_{v_{i}}h\right)  \right\Vert
_{L^{2}\left(  \mu\right)  }^{2}\right)  +3\varpi\int\chi_{R}\left\vert
\nabla_{x}h\cdot\nabla_{v}h\right\vert \mu dxdv\\
&  =\frac{bc_{0}}{8c}t\left\Vert \nabla_{x}h\right\Vert _{L_{\sigma}%
^{2}\left(  \mu\right)  }^{2}+\frac{8c}{bc_{0}t}\left\Vert \nabla
_{v}h\right\Vert _{L_{\sigma}^{2}\left(  \mu\right)  }^{2}+3\varpi\int\chi
_{R}\left\vert \nabla_{x}h\cdot\nabla_{v}h\right\vert \mu dxdv.
\end{align*}
By the Cauchy-Schwartz inequality,
\begin{align*}
\left\vert 2\sum_{i=1}^{3}\int\nabla_{v}\left(  \partial_{x_{i}}h\right)
\cdot\nabla_{v}\left(  \partial_{v_{i}}h\right)  \mu dxdv\right\vert  &
\leq\sum_{i=1}^{3}\left(  \frac{bc_{0}}{8c}t\left\Vert \nabla_{v}\left(
\partial_{x_{i}}h\right)  \right\Vert _{L^{2}\left(  \mu\right)  }^{2}%
+\frac{8c}{bc_{0}t}\left\Vert \nabla_{v}\left(  \partial_{v_{i}}h\right)
\right\Vert _{L^{2}\left(  \mu\right)  }^{2}\right) \\
&  \leq\frac{bc_{0}}{8c}t\left\Vert \nabla_{x}h\right\Vert _{L_{\sigma}%
^{2}\left(  \mu\right)  }^{2}+\frac{8c}{bc_{0}t}\left\Vert \nabla
_{v}h\right\Vert _{L_{\sigma}^{2}\left(  \mu\right)  }^{2},
\end{align*}%
\begin{align*}
&  \quad\left\vert \left(  4ct-2at\right)  \int\nabla_{x}h\cdot\nabla_{v}h\mu
dxdv\right\vert +\int6c\varpi t^{2}\chi_{R}\left\vert \nabla_{x}h\cdot
\nabla_{v}h\right\vert \mu dxdv\\
&  \leq\left(  4c+2a+6c\varpi t\right)  t\left[  \varepsilon t\left\Vert
\nabla_{x}h\right\Vert _{L^{2}\left(  \mu\right)  }^{2}+\frac{C_{\varepsilon}%
}{t}\left\Vert \nabla_{v}h\right\Vert _{L^{2}\left(  \mu\right)  }^{2}\right]
,\ \ \ C_{\varepsilon}=O\left(  \frac{1}{\varepsilon}\right)  ,
\end{align*}
and%
\begin{align*}
&  \quad\left\vert \sum_{i=1}^{3}\int\left(  \left\langle v\right\rangle
^{\gamma-1}\left\vert \nabla_{v}\left(  \partial_{x_{i}}h\right)  \right\vert
\left\vert \partial_{v_{i}}h\right\vert +\left\langle v\right\rangle
^{\gamma-1}\left\vert \nabla_{v}\left(  \partial_{v_{i}}h\right)  \right\vert
\left\vert \partial_{x_{i}}h\right\vert \right)  \mu dxdv\right\vert \\
&  \leq\sum_{i=1}^{3}\left(  \frac{2c}{bc_{0}t}\int\left\langle v\right\rangle
^{2\gamma-2}\left\vert \partial_{v_{i}}h\right\vert ^{2}\mu dxdv+\frac{bc_{0}%
}{8c}t\int\left\vert \nabla_{v}\left(  \partial_{x_{i}}h\right)  \right\vert
^{2}\mu dxdv\right) \\
&  \quad\ +\sum_{i=1}^{3}\left(  \frac{bc_{0}}{8c}t\int\left\langle
v\right\rangle ^{2\gamma-2}\left\vert \partial_{x_{i}}h\right\vert ^{2}\mu
dxdv+\frac{2c}{bc_{0}t}\int\left\vert \nabla_{v}\left(  \partial_{v_{i}%
}h\right)  \right\vert ^{2}\mu dxdv\right) \\
&  \leq\frac{2c}{bc_{0}t}\left\Vert \nabla_{v}h\right\Vert _{L_{\sigma}%
^{2}\left(  \mu\right)  }^{2}+\frac{bc_{0}}{8c}t\left\Vert \nabla
_{x}h\right\Vert _{L_{\sigma}^{2}\left(  \mu\right)  }^{2}.
\end{align*}
In view of $\left(  \ref{*}\right)  $ and $\left(  \ref{x-derivative-mu}%
\right)  $,
\begin{align*}
&  \quad\left\vert \frac{1}{2}\sum_{i=1}^{3}\int\partial_{v_{i}}\left(
\frac{\left\vert \nabla_{v}\Phi\right\vert ^{2}}{4}-\frac{\triangle_{v}\Phi
}{2}+\varpi\chi_{R}\right)  h^{2}\partial_{x_{i}}\mu dxdv\right\vert \\
&  \leq\frac{C}{2}\int\left\langle v\right\rangle ^{2\gamma-3}h^{2}\mu
dxdv+\frac{\varpi}{4}\int\left\langle v\right\rangle ^{2\gamma-3}\chi_{R}%
h^{2}\mu dxdv+\frac{\varpi}{2R}\int\left\langle v\right\rangle ^{2\gamma
-3}\left\vert \chi_{R}^{\prime}\right\vert h^{2}\mu dxdv\\
&  \leq M^{\prime}\int\left\langle v\right\rangle ^{2\gamma-2}h^{2}\mu
dxdv\leq M^{\prime}\left\Vert h\right\Vert _{L_{\sigma}^{2}\left(  \mu\right)
}^{2},
\end{align*}
where $M^{\prime}>0$ is dependent only upon $R$ and $\varpi.$ Gathering the
above estimates, we therefore deduce%
\begin{align*}
\frac{d}{dt}\mathcal{F}\left(  t,h_{t}\right)   &  \leq\left\Vert h\right\Vert
_{L_{\sigma}^{2}\left(  \mu\right)  }^{2}\left[  -c_{0}A+a+6atC_{\varepsilon
}+\left(  4c+2a+6c\varpi t\right)  C_{\varepsilon}+2cM^{\prime}t^{2}\right] \\
&  \quad+\left\Vert \nabla_{x}h\right\Vert _{L^{2}\left(  \mu\right)  }%
^{2}\left(  -2c+3b+\left(  4c+2a+6c\varpi t\right)  \varepsilon\right)
t^{2}\\
&  \quad+\left\Vert \nabla_{x}h\right\Vert _{L_{\sigma}^{2}\left(  \mu\right)
}^{2}\left(  -\frac{bc_{0}}{4}\right)  t^{3}\\
&  \quad+\left\Vert \nabla_{v}h\right\Vert _{L_{\sigma}^{2}\left(  \mu\right)
}^{2}\left(  -ac_{0}+2a\varepsilon+\frac{36c^{2}}{bc_{0}}\right)  t.
\end{align*}
Set $a=\varepsilon,$ $4b=c=\varepsilon^{3/2}.\ $After choosing $A>0$
sufficiently large and $\varepsilon>0$ sufficiently small, we obtain%
\[
\frac{d}{dt}\mathcal{F}\left(  t,h_{t}\right)  \leq0,\ \ \ t\in\left(
0,1\right)  ,
\]
which implies that
\[
\mathcal{F}\left(  t,h_{t}\right)  \leq\mathcal{F}\left(  0,h_{0}\right)
=A\left\Vert h_{0}\right\Vert _{L^{2}}^{2},\ \ \ t\in\left[  0,1\right]  .
\]
This completes the proof of the lemma.
\end{proof}

Before the end of this section, we introduce the wave-remainder decomposition,
which is the key decomposition in our paper. The strategy is to design a
Picard-type iteration, treating $Kf$ as a source term. The zero order
approximation of the Fokker-Planck equation (\ref{in.1.c}) is
\begin{equation}
\left\{
\begin{array}
[c]{l}%
\pa_{t}h^{(0)}+v\cdot\nabla_{x}h^{(0)}+\Lambda h^{(0)}=0\,,\\[4mm]%
h^{(0)}(0,x,v)=f_{0}(x,v)\,.
\end{array}
\right.  \label{bot.3.b}%
\end{equation}
Thus the difference $f-h^{(0)}$ satisfies
\[
\left\{
\begin{array}
[c]{l}%
\pa_{t}(f-h^{(0)})+v\cdot\nabla_{x}(f-h^{(0)})+\Lambda(f-h^{(0)}%
)=K(f-h^{(0)})+Kh^{(0)}\,,\\[4mm]%
(f-h^{(0)})(0,x,v)=0\,.
\end{array}
\right.
\]
Therefore the first order approximation $h^{(1)}$ can be defined by
\begin{equation}
\left\{
\begin{array}
[c]{l}%
\pa_{t}h^{(1)}+v\cdot\nabla_{x}h^{(1)}+\Lambda h^{(1)}=Kh^{(0)}\,,\\[4mm]%
h^{(1)}(0,x,v)=0\,.
\end{array}
\right.  \label{bot.3.c}%
\end{equation}
In general, we can define the $j^{\mathrm{th}}$ order approximation $h^{(j)}$,
$j\geq1$, as
\begin{equation}
\left\{
\begin{array}
[c]{l}%
\pa_{t}h^{(j)}+v\cdot\nabla_{x}h^{(j)}+\Lambda h^{(j)}=Kh^{(j-1)}\,,\\[4mm]%
h^{(j)}(0,x,v)=0\,.
\end{array}
\right.  \label{bot.3.d}%
\end{equation}

The wave part and the remainder part can be defined as follows:
\[
W^{(3)}=\sum_{j=0}^{3}h^{(j)}\,,\quad\mathcal{R}^{(3)}=f-W^{(3)}\,,
\]
$\mathcal{R}^{(3)}$ solving the equation:
\begin{equation}
\left\{
\begin{array}
[c]{l}%
\pa_{t}\mathcal{R}^{(3)}+v\cdot\nabla_{x}\mathcal{R}^{(3)}=L\mathcal{R}%
^{(3)}+Kh^{(3)}\,,\\[4mm]%
\mathcal{R}^{(3)}(0,x,v)=0\,.
\end{array}
\right.
\end{equation}
The next two lemmas are the fundamental properties of $h^{(j)},0\leq j\leq3$.

\begin{lemma}
[$L^{2}$ estimate of $h^{(j)}$, $0\leq j\leq3$]\label{initial-sing} For all
$0\leq j\leq3$ and $t>0$,

\noindent\textrm{(i)} If $\gamma\geq1$, there exists $C>0$ such that
\[
\Vert h^{(j)}\Vert_{L^{2}(\mu)}\lesssim t^{j}e^{-Ct}\Vert f_{0}\Vert
_{L^{2}(\mu)}\,,
\]
\noindent\textrm{(ii)\ }If $0<\gamma<1$, we have
\[
\Vert h^{(j)}\Vert_{L^{2}(\mu)}\lesssim t^{j}\Vert f_{0}\Vert_{L^{2}(\mu)}\,.
\]

\end{lemma}

This lemma is immediate from Lemma \ref{x-energy} and hence we omit the details.

\begin{lemma}
[$x$-derivative estimate of $h^{(j)}$, $0\leq j\leq3$]\label{second-der} Let
$\gamma>0$ and $k=1,2$. Then

\noindent\textrm{(i)} For $0<t\leq1$, we have
\[
\Vert\nabla_{x}^{k}h^{(j)}\Vert_{L^{2}(\mu)}\lesssim t^{j-\frac{3}{2}k}\Vert
f_{0}\Vert_{L^{2}(\mu)}\,.
\]
\noindent\textrm{(ii)} For $t>1$, we have that if $\gamma\geq1$,
\[
\Vert\nabla_{x}^{k}h^{(j)}\Vert_{L^{2}(\mu)}\lesssim t^{j}e^{-Ct}\Vert
f_{0}\Vert_{L^{2}(\mu)}\,,
\]
and if $0<\gamma<1$,
\[
\Vert\nabla_{x}^{k}h^{(j)}\Vert_{L^{2}(\mu)}\lesssim t^{j}\Vert f_{0}%
\Vert_{L^{2}(\mu)}\,.
\]

\end{lemma}

\begin{proof}
We divide our proof into several steps: \newline Step 1: First $x$-derivative
of $h^{(j)}$, $0\leq j\leq3$ in small time. We want to show that for
$0<t\leq1$,
\[
\Vert\nabla_{x}h^{(j)}\Vert_{L^{2}(\mu)}\lesssim t^{(-3+2j)/2}\Vert f_{0}%
\Vert_{L^{2}(\mu)}\,.
\]
The estimate of $h^{(0)}$ is immediate from Lemma \ref{regularization}. Note
that
\[
h^{(1)}=\int_{0}^{t}e^{(t-s)\mathcal{L}}Ke^{s\mathcal{L}}f_{0}ds\,,
\]
hence
\begin{align*}
\nabla_{x}h^{(1)}  &  =\int_{0}^{t}\frac{(t-s)+s}{t}\nabla_{x}%
e^{(t-s)\mathcal{L}}Ke^{s\mathcal{L}}f_{0}ds\,\\
&  =\int_{0}^{t}\frac{1}{t}\left[  (t-s)\nabla_{x}e^{(t-s)\mathcal{L}%
}Ke^{s\mathcal{L}}f_{0}+e^{(t-s)\mathcal{L}}K\left(  s\nabla_{x}%
e^{s\mathcal{L}}f_{0}\right)  \right]  ds.
\end{align*}
From Lemma \ref{x-energy} and Lemma \ref{regularization}, it follows
\begin{align*}
\left\Vert \nabla_{x}h^{(1)}\right\Vert _{L^{2}(\mu)}  &  \lesssim\int_{0}%
^{t}t^{-1}\left[  (t-s)^{-1/2}+s^{-1/2}\right]  ds\Vert f_{0}\Vert_{L^{2}%
(\mu)}\\
&  \lesssim t^{-1/2}\Vert f_{0}\Vert_{L^{2}(\mu)}\,.
\end{align*}
Likewise, note that
\[
h^{(2)}=\int_{0}^{t}\int_{0}^{s_{1}}e^{(t-s_{1})\mathcal{L}}Ke^{(s_{1}%
-s_{2})\mathcal{L}}Ke^{s_{2}\mathcal{L}}f_{0}ds_{2}ds_{1}\,,
\]
and%
\[
\nabla_{x}h^{(2)}=\int_{0}^{t}\int_{0}^{s_{1}}\frac{(s_{1}-s_{2})+s_{2}}%
{s_{1}}\nabla_{x}e^{(t-s_{1})\mathcal{L}}Ke^{(s_{1}-s_{2})\mathcal{L}%
}Ke^{s_{2}\mathcal{L}}f_{0}ds_{2}ds_{1}\,,
\]
hence we have
\begin{align*}
\left\Vert \nabla_{x}h^{(2)}\right\Vert _{L^{2}(\mu)}  &  \lesssim\int_{0}%
^{t}\int_{0}^{s_{1}}s_{1}^{-1}\left[  (s_{1}-s_{2})^{-1/2}+s_{2}%
^{-1/2}\right]  ds_{2}ds_{1}\Vert f_{0}\Vert_{L^{2}(\mu)}\\
&  \lesssim t^{1/2}\Vert f_{0}\Vert_{L^{2}(\mu)}\,.
\end{align*}
The estimate of $h^{(3)}$ is analogous and hence we omit the details. \newline
Step 2: Second $x$-derivatives of $h^{(j)}$, $0\leq j\leq3$ in small time. We
want to show that for any $0<t\leq1$,
\[
\Vert\nabla_{x}^{2}h^{(j)}\Vert_{L^{2}(\mu)}\leq C_{j}t^{-3+j}\Vert f_{0}%
\Vert_{L^{2}(\mu)}\,.
\]
We only give the estimates for $h^{(0)}$ and $h^{(1)},$ and the others are
similar. For any $0<t_{0}\leq1$ and $t_{0}/2<t\leq t_{0}$, we have
\[
\nabla_{x}h^{(0)}(t,x,v)=e^{(t-t_{0}/2)\mathcal{L}}\left[  \nabla_{x}%
h^{(0)}(t_{0}/2,x,v)\right]  .
\]
By Lemma \ref{regularization},
\begin{equation}
\left\Vert \nabla_{x}^{2}h^{(0)}\right\Vert _{L^{2}(\mu)}\left(  t\right)
\lesssim(t-t_{0}/2)^{-3/2}(t_{0}/2)^{-3/2}\Vert f_{0}\Vert_{L^{2}(\mu)}\,.
\label{f0-xx}%
\end{equation}
Taking $t=t_{0}$ yields
\begin{equation}
\left\Vert \nabla_{x}^{2}h^{(0)}\right\Vert _{L^{2}(\mu)}(t_{0})\lesssim
t_{0}^{-3}\Vert f_{0}\Vert_{L^{2}(\mu)}\,. \label{f0-xx}%
\end{equation}
Since $t_{0}\in(0,1]$ is arbitrary, this completes the estimate for $h^{(0)}$.
For $0<t_{1}\leq1$ and $t_{1}/2<t\leq t_{1}$, we have
\[
\nabla_{x}h^{(1)}(t,x,v)=e^{(t-t_{1}/2)\mathcal{L}}\left[  \nabla_{x}%
h^{(1)}(t_{1}/2,x,v)\right]  +\int_{t_{1}/2}^{t}e^{(t-s)\mathcal{L}}\left[
K\nabla_{x}h^{(0)}(s,x,v)\right]  ds\,.
\]
Using Lemma \ref{regularization} and (\ref{f0-xx}) gives
\[
\left\Vert \nabla_{x}^{2}h^{(1)}\right\Vert _{L^{2}(\mu)}(t)\lesssim
(t-t_{1}/2)^{-3/2}(t_{1}/2)^{-1/2}\Vert f_{0}\Vert_{L^{2}(\mu)}+\int_{t_{1}%
/2}^{t}s^{-3}\Vert f_{0}\Vert_{L^{2}(\mu)}ds\,.
\]
Taking $t=t_{1}$, we get
\[
\left\Vert \nabla_{x}^{2}h^{(1)}\right\Vert _{L^{2}(\mu)}(t_{1})\lesssim
t_{1}^{-2}\Vert f_{0}\Vert_{L^{2}(\mu)}\,.
\]
Since $t_{1}\in(0,1]$ is arbitrary, this completes the estimate for $h^{(1)}$.

Next, we shall prove the large time behavior for $\gamma\geq1$; the estimate
for the case $0<\gamma<1$ can be obtained by employing the same argument.
\newline Step 3: First $x$-derivative of $h^{(j)}$, $0\leq j\leq3,$ in large
time for $\gamma\geq1$. We want to show that for $t>1$,
\[
\Vert\nabla_{x}h^{(j)}\Vert_{L^{2}(\mu)}\lesssim t^{j}e^{-Ct}\Vert f_{0}%
\Vert_{L^{2}(\mu)}\,,\ \ 0\leq j\leq3.
\]
In view of Lemma \ref{x-energy}, we have
\begin{equation}
\left\Vert \nabla_{x}h^{(0)}\right\Vert _{L^{2}(\mu)}(t)\leq e^{-C(t-1)}%
\left\Vert \nabla_{x}h^{(0)}\right\Vert _{L^{2}(\mu)}(1)\lesssim e^{-Ct}\Vert
f_{0}\Vert_{L^{2}(\mu)}\,,\ \ t>1. \label{h0}%
\end{equation}
For $h^{(1)}$, we have
\[
h^{(1)}(t,x,v)=e^{(t-1)\mathcal{L}}h^{(1)}(1,x,v)+\int_{1}^{t}%
e^{(t-s)\mathcal{L}}\left[  Kh^{(0)}(s,x,v)\right]  ds\,,\ \ t>1.
\]
Using Lemma \ref{x-energy} and (\ref{h0}) gives
\begin{align*}
\left\Vert \nabla_{x}h^{(1)}\right\Vert _{L^{2}(\mu)}(t)  &  \leq
e^{-C(t-1)}\left\Vert \nabla_{x}h^{(1)}\right\Vert _{L^{2}(\mu)}\left(
1\right)  +\int_{1}^{t}e^{-C(t-s)}\left\Vert \nabla_{x}h^{(0)}\right\Vert
_{L^{2}(\mu)}(s)ds\\
&  \lesssim te^{-Ct}\Vert f_{0}\Vert_{L^{2}(\mu)}\,,\ \ \ \ \ t>1,
\end{align*}
and similarly for $\nabla_{x}h^{(2)}$ and $\nabla_{x}h^{(3)}.$\newline Step 4:
Second $x$-derivatives of $h^{(j)}$, $0\leq j\leq3$ in large time for
$\gamma\geq1$. We demonstrate that for $t>1$,
\[
\Vert\nabla_{x}^{2}h^{(j)}\Vert_{L^{2}(\mu)}\lesssim t^{j}e^{-Ct}\Vert
f_{0}\Vert_{L^{2}(\mu)}\,,
\]
whose proof is similar to Step 3.
\end{proof}


\section{In the time-like region}

\label{large-time}

In this section, we will see the large time behavior of solutions to equation
(\ref{in.1.c}). In the sequel, we separate our discussion for the case
$\gamma\geq1$ and the case $0<\gamma<1$.

\subsection{The case $\gamma\geq1$}

According to Lemma \ref{second-der}, together with the Sobolev inequality
\cite[Theorem 5.8]{[Adams]} :%
\[
\left\Vert f\right\Vert _{L_{x}^{\infty}L_{v}^{2}}\lesssim\left\Vert
\nabla_{x}^{2}f\right\Vert _{L^{2}}^{3/4}\left\Vert f\right\Vert _{L^{2}%
}^{1/4},
\]
we immediate obtain the behavior of the wave part as follows.

\begin{proposition}
\label{pointwise-wave part >=1}Assume that $\gamma\geq1$. Then for $0\leq
j\leq3$ and $t>0$, there exists $C>0$ such that
\[
|h^{(j)}|_{L_{v}^{2}}\lesssim e^{-Ct}t^{j-\frac{9}{4}}\Vert f_{0}\Vert_{L^{2}%
}\,.
\]

\end{proposition}

Based on the wave-remainder decomposition, it remains to study the large time
behavior of the remainder part. By the Fourier transform with respect to the
$x$ variable, the solution of the Fokker-Planck equation (\ref{in.1.c}) can be
represented as
\begin{equation}
\displaystyle f(t,x,v)=\int_{{\mathbb{R}}^{3}}e^{i\eta\cdot x+(-iv\cdot
\eta+L)t}\widehat{f}_{0}(\eta,v)d\eta\,.
\end{equation}
We can decompose the solution $f$ into the long wave part $f_{L}$ and the
short wave part $f_{S}$ given respectively by
\begin{equation}%
\begin{array}
[c]{l}%
\label{bot.2.e}\displaystyle f_{L}=\int_{|\eta|<\de}e^{i\eta\cdot
x+(-iv\cdot\eta+L)t}\widehat{f}_{0}(\eta,v)d\eta\,,\\
\\
\displaystyle f_{S}=\int_{|\eta|>\de}e^{i\eta\cdot x+(-iv\cdot\eta
+L)t}\widehat{f}_{0}(\eta,v)d\eta\,.
\end{array}
\end{equation}
The following short wave analysis relies on spectral analysis (Lemma
\ref{pr12}).

\begin{proposition}
[Short wave $f_{S}$]Assume that $\gamma\geq1$\ and $f_{0}\in L^{2}$. Then
\begin{equation}
\Vert f_{S}\Vert_{L^{2}}\lesssim e^{-a(\tau)t}\Vert f_{0}\Vert_{L^{2}}\,.
\label{bot.2.f}%
\end{equation}

\end{proposition}

In order to study the long wave part $f_{L}$ for $\gamma\geq1$, we need to
further decompose the long wave part as the fluid part and the nonfluid part,
i.e., $f_{L}=f_{L;0}+f_{L;\perp}$, where
\begin{equation}%
\begin{array}
[c]{l}%
\label{bot.2.g}\displaystyle f_{L;0}=\int_{|\eta|<\de}e^{\lambda(\eta
)t}e^{i\eta\cdot x}\big<e_{D}(-\eta),\hat{f_{0}}\big>_{v}e_{D}(\eta)d\eta\,,\\
\\
\displaystyle f_{L;\perp}=\int_{|\eta|<\de}e^{i\eta\cdot x}e^{(-iv\cdot
\eta+L)t}\Pi_{\eta}^{D\perp}\hat{f}_{0}d\eta\,.
\end{array}
\end{equation}

Using Lemma \ref{pr12}, we obtain the exponential decay of the nonfluid long
wave part.

\begin{proposition}
[Non fluid long wave $f_{L;\perp}$]Assume that $\gamma\geq1$ and $f_{0}\in
L^{2}$. Then for $s>0$,
\begin{equation}
\Vert f_{L;\perp}\Vert_{H_{x}^{s}L_{v}^{2}}\lesssim e^{-a(\tau)t}\Vert
f_{0}\Vert_{L^{2}}\,. \label{bot.2.h}%
\end{equation}

\end{proposition}

For the fluid part, we have the following structure:

\begin{proposition}
[Fluid long wave $f_{L;0}$]\label{fluid} For $\ga\geq3/2$ and any given $M>1,$
there exists $C>0$ such that for $|x|\leq Mt$,
\begin{equation}
\left\vert f_{L;0}(t,x,v)\right\vert _{L_{v}^{2}}\leq C\left[  (1+t)^{-3/2}%
e^{-\frac{|x|^{2}}{C(t+1)}}+e^{-t/C}\right]  \Vert f_{0}\Vert_{L_{x}^{1}%
L_{v}^{2}}\,.
\end{equation}
On the other hand, for $1\leq\ga<3/2$ and any given positive integer $N$,
there exists a positive constant $C$ depending on $N$ such that
\begin{equation}
\left\vert f_{L;0}(t,x,v)\right\vert _{L_{v}^{2}}\leq C\left[  (1+t)^{-3/2}%
\Big(1+\frac{|x|^{2}}{1+t}\Big)^{-N}+e^{-t/C}\right]  \Vert f_{0}\Vert
_{L_{x}^{1}L_{v}^{2}}\,.
\end{equation}

\end{proposition}

\begin{proof}
Before the proof of this proposition, we need the following two lemmas:

\begin{lemma}
[{Lemma 7.11, \cite{[LiuYu1]}}]\label{heatkernel_ana} Suppose that
$g(t,\eta,v)$ is analytic in $\eta$ for $|\eta|<\delta\ll1$ and satisfies
\[
|g(t,\eta,v)|_{L_{v}^{2}}\lesssim e^{-A|\eta|^{2}t+O(|\eta|^{4})t}\,,
\]
for some constant $A>0$. Then in the region of $|x|<(\mathfrak{M}+1)t$, where
$\mathfrak{M}$ is any given positive constant, there exists a constant $C>0$
such that the following inequality holds:
\[
\displaystyle\left\vert \int_{|\eta|<\delta}e^{ix\cdot\eta}g(t,\eta
,v)d\eta\right\vert _{L_{v}^{2}}\leq C\left[  (1+t)^{-\frac{3}{2}}%
e^{-\frac{|x|^{2}}{Ct}}+e^{-t/C}\right]  .
\]

\end{lemma}

\begin{lemma}
[{Lemma 2.2, \cite{[LiuWang]}}]\label{heatkernel_smo} Let $x,\eta
,v\in\mathbb{R}^{3}$. Suppose $g(t,\eta,v)$ is smooth and has compact support
in the variable $\eta$, and there exists a constant $b>0$ such that
$g(t,\eta,v)$ satisfies
\[
\left\vert D_{\eta}^{\beta}(g(t,\eta,v))\right\vert _{L_{v}^{2}}\leq C_{\beta
}(1+t)^{|\beta|/2}e^{-b|\eta|^{2}t},
\]
for any multi-indexes $\beta$ with $|\beta|\leq2N$. Then there exists positive
constants $C_{N}$ such that
\[
\left\vert \int_{{\mathbb{R}}^{3}}e^{ix\cdot\eta}g(t,\eta,v)d\eta\right\vert
_{L_{v}^{2}}\leq C_{N}\left[  (1+t)^{-3/2}B_{N}(|x|,t)\right]  ,
\]
where $N$ is any fixed integer and
\[
B_{N}(|x|,t)=\left(  1+\frac{|x|^{2}}{1+t}\right)  ^{-N}.
\]

\end{lemma}

We now return to the proof of this proposition. Notice that
\[
f_{L;0}(t,x,v)=\int_{|\eta|<\de}e^{i\eta\cdot x}e^{\lambda(\eta)t}%
\big<e_{D}(-\eta),\hat{f_{0}}\big>_{v}e_{D}(\eta)d\eta\,.
\]
Let
\[
g(t,\eta,v)=e^{\lambda(\eta)t}\big<e_{D}(-\eta),\hat{f_{0}}\big>_{v}e_{D}%
(\eta)\cdot\mathbf{1}_{\left\{  \left\vert \eta\right\vert <\delta\right\}
},
\]
where $\mathbf{1}_{D}$ is the characteristic function of the domain $D$. When
$\gamma\geq3/2$, the eigenvalue $\lambda(\eta)$ and eigenvector $e_{D}(\eta)$
are analytic in $\eta$. Owing to the asymptotic expansion of $\lambda(\eta)$
in \eqref{pre.ab.h},\ we have
\[
|g(t,\eta,v)|_{L_{v}^{2}}\leq e^{-a_{\ga}|\eta|^{2}t+O(|\eta|^{4})t}\left\Vert
f_{0}\right\Vert _{L_{x}^{1}L_{v}^{2}}.
\]
From Lemma \ref{heatkernel_ana} it follows
\[
\left\vert f_{L;0}(t,x,v)\right\vert _{L_{v}^{2}}\lesssim\left[
(1+t)^{-3/2}e^{-\frac{|x|^{2}}{C(t+1)}}+e^{-t/C}\right]  \left\Vert
f_{0}\right\Vert _{L_{x}^{1}L_{v}^{2}}\,.
\]
As for $1\leq\gamma<3/2$, the eigenvalue and eigenvector are only smooth in
$\eta$. In this case, one can see that
\[
\left\vert D_{\eta}^{\beta}g(t,\eta,v)\right\vert _{L_{v}^{2}}\lesssim\left(
1+t^{\left\vert \beta\right\vert /2}\right)  \left(  1+\left\vert
\eta\right\vert ^{2}t\right)  ^{\left\vert \beta\right\vert /2}e^{-a_{\gamma
}\left\vert \eta\right\vert ^{2}t/2}\left\Vert f_{0}\right\Vert _{L_{x}%
^{1}L_{v}^{2}},
\]
since $f_{0}$ has compact support in the $x$ variable. Note that the
polynomial growth $\left(  1+\left\vert \eta\right\vert ^{2}t\right)
^{\left\vert \beta\right\vert /2}$ can be absorbed by the exponential decay,
hence we can conclude that%
\[
\left\vert f_{L;0}(t,x,v)\right\vert _{L_{v}^{2}}\lesssim\left[
(1+t)^{-3/2}\Big(1+\frac{|x|^{2}}{1+t}\Big)^{-N}\right]  \Vert f_{0}%
\Vert_{L_{x}^{1}L_{v}^{2}}\,,
\]
in accordance with Lemma \ref{heatkernel_smo}.
\end{proof}

We define the fluid part as $f_{F}=f_{L,0}$ and the nonfluid part as $f_{\ast
}=f-f_{F}=f_{L;\perp}+f_{S}$. By the fluid-nonfluid decomposition and the
wave-remainder decomposition, we have
\[
f=f_{F}+f_{\ast}=W^{(3)}+\mathcal{R}^{(3)}\,.
\]
We can define the tail part as $f_{R}=\mathcal{R}^{(3)}-f_{F}$ and so $f$ can
be written as $f=W^{(3)}+f_{F}+f_{R}.$

It follows from Proposition \ref{pointwise-wave part >=1} and \ref{fluid} that
the estimates of wave part $W^{(3)}$ and the fluid part $f_{F}$ inside the
time-like region is completed. Hence, it remains to study the tail part
$f_{R}$. From Lemma \ref{second-der} and the fact that {$\mathbb{G}^{t},K$ are
bounded operators on $L^{2},$ }$\mathcal{R}^{(3)}$ has the following estimate
\begin{equation}
\Vert\mathcal{R}^{(3)}(t)\Vert_{H_{x}^{2}L_{v}^{2}}\leq\int_{0}^{t}\left\Vert
h^{(3)}(s)\right\Vert _{H_{x}^{2}L_{v}^{2}}ds\lesssim\Vert f_{0}\Vert_{L^{2}%
}\,. \label{R3}%
\end{equation}
In view of (\ref{bot.2.f}), (\ref{bot.2.h}), (\ref{R3}) and Lemma
\ref{initial-sing}, there exists $C>0$ such that
\[
\Vert f_{R}\Vert_{L^{2}}=\Vert f_{\ast}-W^{(3)}\Vert_{L^{2}}\lesssim
e^{-Ct}\Vert f_{0}\Vert_{L^{2}}\,,
\]
and%
\[
\Vert f_{R}\Vert_{H_{x}^{2}L_{v}^{2}}=\Vert\mathcal{R}^{(3)}-f_{F}\Vert
_{H_{x}^{2}L_{v}^{2}}\lesssim\Vert f_{0}\Vert_{L^{2}}\,.
\]
The Sobolev inequality \cite[Theorem 5.8]{[Adams]} implies
\begin{equation}
\left\vert f_{R}\right\vert _{L_{v}^{2}}\leq\Vert f_{R}\Vert_{L_{x}^{\infty
}L_{v}^{2}}\lesssim\Vert f_{R}\Vert_{H_{x}^{2}L_{v}^{2}}^{3/4}\Vert f_{R}%
\Vert_{L^{2}}^{1/4}\lesssim e^{-Ct}\Vert f_{0}\Vert_{L^{2}}\,, \label{fR}%
\end{equation}
for some constant $C>0$. In conclusion, we have that for the time-like region,
if $\gamma\geq3/2$, there exists a constant $C>0$ such that
\begin{equation}
\left\vert \mathcal{R}^{(3)}\right\vert _{L_{v}^{2}}\lesssim\left[
(1+t)^{-3/2}e^{-C\frac{|x|^{2}}{t+1}}+e^{-Ct}\right]  \Vert f_{0}\Vert
_{L_{x}^{\infty}L_{v}^{2}}\,; \label{R-larger than 3/2}%
\end{equation}
and if $1\leq\gamma<3/2$, any given $N>0$, there exists a constant $C>0$ such
that
\begin{equation}
\left\vert \mathcal{R}^{(3)}\right\vert _{L_{v}^{2}}\lesssim\left[
(1+t)^{-3/2}\Big(1+\frac{|x|^{2}}{1+t}\Big)^{-N}+e^{-Ct}\right]  \Vert
f_{0}\Vert_{L_{x}^{\infty}L_{v}^{2}}\,. \label{R-less than 3/2}%
\end{equation}

Combining Proposition \ref{pointwise-wave part >=1}, (\ref{fR}),
(\ref{R-larger than 3/2}) and (\ref{R-less than 3/2}), we obtain the pointwise
estimate for the solution in the time-like region.

\begin{theorem}
[Time-like region for $\gamma\geq1$]%
\label{time-like region for gamma larger or equal than 1}Let $\gamma\geq1$ and
let $f$ be the solution to equation (\ref{in.1.c}). Assume that the initial
condition $f_{0}$ has compact support in the $x$ variable and is bounded in
$L_{v}^{2}$. Then for any given $M>1$ and $|x|\leq Mt$,

\noindent\textrm{(i)} As $\gamma\geq3/2$, there exists a positive constant $C$
such that
\begin{equation}
\left\vert f\right\vert _{L_{v}^{2}}\lesssim\left[  (1+t)^{-3/2}%
e^{-C\frac{|x|^{2}}{t+1}}+(1+t^{-9/4})e^{-Ct}\right]  \Vert f_{0}\Vert
_{L_{x}^{\infty}L_{v}^{2}}\,;
\end{equation}
\newline\textrm{(ii) }As $1\leq\gamma<3/2$, any given $N>0$, there exists a
constant $C>0$ such that
\begin{equation}
\left\vert f\right\vert _{L_{v}^{2}}\lesssim\left[  (1+t)^{-3/2}%
\Big(1+\frac{|x|^{2}}{1+t}\Big)^{-N}+(1+t^{-9/4})e^{-Ct}\right]  \Vert
f_{0}\Vert_{L_{x}^{\infty}L_{v}^{2}}\,.
\end{equation}

\end{theorem}

\subsection{The case $0<\gamma<1$}

First, we introduce the $L^{2}$ estimate and the pointwise estimate of the
wave part.

\begin{proposition}
\label{pointwise-wave part <1}Assume that $0<\gamma<1$. Then for $0\leq
j\leq3,$ and $t>0$, there exists $c_{\gamma}>0$ such that
\begin{equation}
\left\Vert h^{(j)}\right\Vert _{L^{2}}\lesssim t^{j}e^{-c_{\gamma}%
\alpha^{\frac{2\left(  1-\gamma\right)  }{2-\gamma}}t^{\frac{\gamma}{2-\gamma
}}}\Vert f_{0}\Vert_{L^{2}\left(  e^{\left(  j+1\right)  \alpha\left\langle
v\right\rangle ^{\gamma}}\right)  }\,, \label{L2-esti-gamma<1}%
\end{equation}
and%
\begin{equation}
\left\vert h^{(j)}\left(  t,x,v\right)  \right\vert _{L_{v}^{2}}\lesssim
t^{j-\frac{9}{4}}e^{-\frac{1}{4}c_{\gamma}\alpha^{\frac{2\left(
1-\gamma\right)  }{2-\gamma}}t^{\frac{\gamma}{2-\gamma}}}\Vert f_{0}%
\Vert_{L^{2}\left(  e^{\left(  j+1\right)  \alpha\left\langle v\right\rangle
^{\gamma}}\right)  }.\, \label{Wave-ptw-gamma<1}%
\end{equation}

\end{proposition}

\begin{proof}
We first consider the $L^{2}$ estimate for $h^{\left(  0\right)  }.\ $It is
easy to see that (\ref{Energy-evolu}) is still valid if setting $\mu\left(
t,x,v\right)  =e^{\alpha\left\langle v\right\rangle ^{\gamma}}$, namely,
\begin{equation}
\frac{d}{dt}\left\Vert h^{\left(  0\right)  }\right\Vert _{L^{2}}%
^{2}+C\left\Vert h^{\left(  0\right)  }\right\Vert _{L_{\sigma}^{2}}^{2}\leq0,
\label{3}%
\end{equation}
and
\begin{equation}
\frac{d}{dt}\left\Vert e^{\alpha\left\langle v\right\rangle ^{\gamma}%
}h^{\left(  0\right)  }\right\Vert _{L^{2}}^{2}+C\left\Vert e^{\alpha
\left\langle v\right\rangle ^{\gamma}}h^{\left(  0\right)  }\right\Vert
_{L_{\sigma}^{2}}^{2}\leq0. \label{4}%
\end{equation}
Hence, $\left\Vert h^{\left(  0\right)  }\right\Vert _{L^{2}}\leq\Vert
f_{0}\Vert_{L^{2}}\ $for $t\geq0$ and it suffices to show that for $t\geq1,$
\[
\left\Vert h^{(0)}\right\Vert _{L^{2}}\lesssim e^{-c_{\gamma}\alpha
^{\frac{2\left(  1-\gamma\right)  }{2-\gamma}}t^{\frac{\gamma}{2-\gamma}}%
}\Vert f_{0}\Vert_{L^{2}\left(  e^{\alpha\left\langle v\right\rangle ^{\gamma
}}\right)  }.
\]
As in the work of Caflisch \cite{[Caflisch]}, we consider a time-dependent low
velocity part
\[
E=\{\left\langle v\right\rangle \leq\beta t^{p^{\prime}}\},
\]
and its complementary high velocity part $E^{c}=\{\left\langle v\right\rangle
>\beta t^{p^{\prime}}\},$ where $p^{\prime}>0$ and $\beta>0$ will be
determined later. Following the argument as in Section 5 of
\cite{[Strain-Guo]}, together with (\ref{3}) and (\ref{4}), we obtain%
\[
\left\Vert h^{\left(  0\right)  }\right\Vert _{L^{2}}\lesssim e^{-c_{\gamma
}\alpha^{\frac{2\left(  1-\gamma\right)  }{2-\gamma}}t^{\frac{\gamma}%
{2-\gamma}}}\left\Vert f_{0}\right\Vert _{L^{2}(e^{\alpha\left\langle
v\right\rangle ^{\gamma}})}\,,
\]
for some constant $c_{\gamma}>0,$ after choosing $p^{\prime}=\frac{1}%
{2-\gamma}$ in the Fokker -Planck case and $\beta>0$ sufficiently large. This
completes the $L^{2}$ estimate for $h^{\left(  0\right)  }$.

Through the Duhumel Principle, we immediately obtain (\ref{L2-esti-gamma<1})
for $1\leq j\leq3.$ Furthermore, the Sobolev inequality, together with Lemma
\ref{second-der} and (\ref{L2-esti-gamma<1}), implies the desired pointwise
estimate for the wave part $h^{\left(  j\right)  },\ 0\leq j\leq3$.
\end{proof}

Next, we are concerned with the pointwise behavior of the remainder part in
the time-like region. In virtue of the lack of the spectral analysis for
$0<\gamma<1$, we will instead use the method of the weighted $L^{2}$ estimate
in the Fourier transformed variable and the interpolation argument to deal
with the time decay of the solution $f$ to equation $\left(  \ref{in.1.c}%
\right)  $ in this case. The main idea is to construct the desired weighted
time-frequency Lyapunov functional to capture the total energy dissipation
rate. In the course of the proof we have to take great care to estimate the
microscopic and macroscopic parts for $\left\vert \eta\right\vert \leq1$ and
$\left\vert \eta\right\vert >1$ respectively. Consider $\left(  \ref{in.1.c}%
\right)  ,$ taking the Fourier transform with respect to the $x$ variable
leads to
\begin{equation}
\partial_{t}\widehat{f}+iv\cdot\eta\widehat{f}=L\widehat{f}. \label{FT}%
\end{equation}
We first calculate the $L^{2}$ estimate.

\begin{proposition}
[$L^{2}$ estimate]Let $f$ be the solution to equation $\left(  \ref{in.1.c}%
\right)  $ . Then there exists a time-frequency functional $\mathcal{E}\left(
t,\eta\right)  $ such that
\begin{equation}
\mathcal{E}\left(  t,\eta\right)  \approx\left\vert \widehat{f}\left(
t,\eta,v\right)  \right\vert _{L_{v}^{2}}^{2},
\end{equation}
where for any $t>0$ and $\eta\in\mathbb{R}^{3}$, we have
\begin{equation}
\partial_{t}\mathcal{E}\left(  t,\eta\right)  +\sigma\widehat{\rho}\left(
\eta\right)  \left\vert \widehat{f}\left(  t,\eta,v\right)  \right\vert
_{L_{\gamma-1}^{2}}^{2}\leq0. \label{E-N}%
\end{equation}
Here we use the notation $\widehat{\rho}\left(  \eta\right)  :=\min
\{1,\left\vert \eta\right\vert ^{2}\}.$
\end{proposition}

\begin{proof}
We multiply equation $\left(  \ref{FT}\right)  $ by $\overline{\widehat
{f}\left(  t,\eta,v\right)  }$ and integrate over $v$ to obtain%
\[
\frac{1}{2}\frac{d}{dt}\left\vert \widehat{f}\left(  t,\eta,v\right)
\right\vert _{L_{v}^{2}}^{2}-\operatorname{Re}\left\langle L\widehat
{f},\widehat{f}\right\rangle =0
\]
From the coercivity in Lemma \ref{co}, it follows that%
\begin{equation}
\frac{1}{2}\frac{d}{dt}\left\vert \widehat{f}\left(  t,\eta,v\right)
\right\vert _{L_{v}^{2}}^{2}+\nu_{0}\left\vert \mathrm{P}_{1}\widehat
{f}\right\vert _{L_{\sigma}^{2}}^{2}\leq0. \label{Fourier-energy -wo-weigh}%
\end{equation}
Now, we need the estimate of $\mathrm{P}_{0}\widehat{f}$. In the sequel, we
will apply Strain's argument to estimate the macroscopic dissipation, in the
spirit of Kawashima's work on dissipation of the hyperbolic-parabolic system.
Let $a=\left\langle \mathcal{M}^{1/2},f\right\rangle _{v}$ and $b=\left(
b_{1},b_{2},b_{3}\right)  \ $with $b_{i}=\left\langle v_{i}\mathcal{M}%
^{1/2},f\right\rangle _{v}=\left\langle v_{i}\mathcal{M}^{1/2},\mathrm{P}%
_{1}f\right\rangle _{v}.$ Then $P_{0}f=a\mathcal{M}^{1/2}$ and from $\left(
\ref{in.1.c}\right)  ,$ $a$ and $b$ satisfy the fuild-type system%
\begin{equation}
\left\{
\begin{array}
[c]{l}%
\partial_{t}a+\nabla_{x}\cdot b=0\vspace{3mm}\\
\partial_{t}b+\alpha\nabla_{x}a+\nabla_{x}\cdot\Gamma\left(  \mathrm{P}%
_{1}f\right)  =-\int\left(  \mathcal{M}^{1/2}\nabla_{v}\Phi\right)
\mathrm{P}_{1}fdv,
\end{array}
\right.  \label{Evolu-ab}%
\end{equation}
where%
\[
\alpha=\frac{1}{3}\int\left\vert v\right\vert ^{2}\mathcal{M}dv>0,
\]
and $\Gamma=\left(  \Gamma_{ij}\right)  _{3\times3}$ is the moment function
defined by
\[
\Gamma_{ij}\left(  g\right)  =\left\langle \left(  v_{i}v_{j}-1\right)
\mathcal{M}^{1/2},g\right\rangle _{v},\ \ \ \ 1\leq i,\ j\leq3.
\]
Note by the definition of $P_{0}$ that $\Gamma\left(  \mathrm{P}_{1}f\right)
=\int\left(  v\otimes v\right)  \mathcal{M}^{1/2}\mathrm{P}_{1}fdv.$ Taking
the Fourier transform with respect to $x$ of $\left(  \ref{Evolu-ab}\right)
,$ we have%
\begin{align*}
\left\vert \eta\right\vert ^{2}\left\vert \widehat{a}\right\vert ^{2}  &
=\left(  i\eta\widehat{a},i\eta\widehat{a}\right)  =\frac{1}{\alpha}\left(
i\eta\widehat{a},-\partial_{t}\widehat{b}-i\Gamma\left(  \mathrm{P}%
_{1}\widehat{f}\right)  \eta-\int\left(  \mathcal{M}^{1/2}\nabla_{v}%
\Phi\right)  \mathrm{P}_{1}\widehat{f}dv\right) \\
&  =\frac{1}{\alpha}\left[  -\left(  i\eta\widehat{a},\widehat{b}\right)
_{t}+\left|  \eta\cdot\widehat{b}\right|  ^{2}-\left(  i\eta\widehat
{a},i\Gamma\left(  \mathrm{P}_{1}\widehat{f}\right)  \eta\right)  -\left(
i\eta\widehat{a},\int\left(  \mathcal{M}^{1/2}\nabla_{v}\Phi\right)
\mathrm{P}_{1}\widehat{f}dv\right)  \right]  .
\end{align*}
\newline Invoking on the rapid decay of $\mathcal{M}^{1/2}$ and using the
Cauchy-Schwartz inequality, we have
\[
\left\vert \int\left(  \mathcal{M}^{1/2}\nabla_{v}\Phi\right)  \mathrm{P}%
_{1}\widehat{f}dv\right\vert ^{2}\leq\left\vert \mathcal{M}^{1/2}v\left\langle
v\right\rangle ^{-1}\right\vert _{L_{v}^{2}}^{2}\left\vert \left\langle
v\right\rangle ^{\gamma-1}\mathrm{P}_{1}\widehat{f}\right\vert _{L_{v}^{2}%
}^{2}\leq3\alpha\left\vert \mathrm{P}_{1}\widehat{f}\right\vert _{L_{\gamma
-1}^{2}}^{2},
\]
and
\[
\left\vert \left(  i\eta\widehat{a},i\Gamma\left(  \mathrm{P}_{1}\widehat
{f}\right)  \eta\right)  \right\vert \leq\epsilon\left\vert \eta\right\vert
^{2}\left\vert \widehat{a}\right\vert ^{2}+C_{\epsilon}\left\vert
\eta\right\vert ^{2}\left\vert \mathrm{P}_{1}\widehat{f}\right\vert
_{L_{\gamma-1}^{2}}^{2},
\]
for any small $\epsilon>0.$ Therefore, we can conclude%
\begin{equation}
\partial_{t}\operatorname{Re}\frac{\left(  i\eta\widehat{a},\widehat
{b}\right)  }{1+\left\vert \eta\right\vert ^{2}}+\frac{\sigma\left\vert
\eta\right\vert ^{2}}{1+\left\vert \eta\right\vert ^{2}}\left\vert \widehat
{a}\right\vert ^{2}\leq C\left\vert \mathrm{P}_{1}\widehat{f}\right\vert
_{L_{\gamma-1}^{2}}^{2}, \label{Fourier energy est.}%
\end{equation}
for some $\sigma>0.$ Now, we define
\begin{equation}
\mathcal{E}\left(  t,\eta\right)  =\left\vert \widehat{f}\left(
t,\eta,v\right)  \right\vert _{L_{v}^{2}}^{2}+\kappa_{3}\operatorname{Re}%
\frac{\left(  i\eta\widehat{a},\widehat{b}\right)  }{1+\left\vert
\eta\right\vert ^{2}}, \label{Lyapu 1}%
\end{equation}
for a constant $\kappa_{3}>0$ to be determined later. One can fix $\kappa
_{3}>0$ small enough such that $\mathcal{E}\left(  t,\eta\right)
\approx\left\vert \widehat{f}\left(  t,\eta,v\right)  \right\vert _{L_{v}^{2}%
}^{2}.$ Furthermore, according to Lemma \ref{co} and $\left(
\ref{Fourier energy est.}\right)  ,$ we choose $\kappa_{3}>0$ sufficiently
small such that%
\begin{equation}
\partial_{t}\mathcal{E}\left(  t,\eta\right)  +\sigma\left\vert \mathrm{P}%
_{1}\widehat{f}\right\vert _{L_{\gamma-1}^{2}}^{2}+\frac{2\sigma\left\vert
\eta\right\vert ^{2}}{1+\left\vert \eta\right\vert ^{2}}\left\vert \widehat
{a}\right\vert ^{2}\leq0, \label{E}%
\end{equation}
for some $\sigma>0.$ In conclusion, we now have
\[
\partial_{t}\mathcal{E}\left(  t,\eta\right)  +\sigma\widehat{\rho}\left(
\eta\right)  \left\vert \widehat{f}\left(  t,\eta,v\right)  \right\vert
_{L_{\gamma-1}^{2}}^{2}\leq0.
\]
Here we use the notation $\widehat{\rho}\left(  \eta\right)  :=\min
\{1,\left\vert \eta\right\vert ^{2}\}$.
\end{proof}

Since $\gamma-1<0,$ it is insufficient to gain the time decay of the total
energy of the solution $f.$ Therefore, in order to capture the total energy
dissipation rate, we need to make further energy estimates on the microscopic
part $\mathrm{P}_{1}f$ and the macroscopic part $\mathrm{P}_{0}f$.

\begin{proposition}
Let $f$ be the solution to equation $\left(  \ref{in.1.c}\right)  $ . Then
there exists a weighted time-frequency functional $\widetilde{\mathcal{E}%
}\left(  t,\eta\right)  $ such that
\begin{equation}
\widetilde{\mathcal{E}}\left(  t,\eta\right)  \approx\left\vert e^{\frac
{\alpha}{2}\left\langle v\right\rangle ^{\gamma}}\widehat{f}\left(
t,\eta,v\right)  \right\vert _{L_{v}^{2}}^{2},
\end{equation}
where $0<\alpha\gamma<1/20\ $and for any $t>0$ and $\eta\in\mathbb{R}^{3}$ we
have
\begin{equation}
\partial_{t}\widetilde{\mathcal{E}}\left(  t,\eta\right)  +\sigma\widehat
{\rho}\left(  \eta\right)  \left\vert e^{\frac{\alpha}{2}\left\langle
v\right\rangle ^{\gamma}}\widehat{f}\left(  t,\eta,v\right)  \right\vert
_{L_{\gamma-1}^{2}}^{2}\leq0. \label{EL}%
\end{equation}
Here we use the notation $\widehat{\rho}\left(  \eta\right)  :=\min\{1,$
$\left\vert \eta\right\vert ^{2}\}.$
\end{proposition}

\begin{proof}
Firstly, we shall prove the following Lyapunov inequality with a velocity
weight $e^{\alpha\left\langle v\right\rangle ^{\gamma}},\ 0<\alpha\gamma
<1/20$:%
\begin{equation}
\frac{d}{dt}\left\vert e^{\frac{\alpha}{2}\left\langle v\right\rangle
^{\gamma}}\mathrm{P}_{1}\widehat{f}\left(  t,\eta,v\right)  \right\vert
_{L_{v}^{2}}^{2}+\sigma\left\vert e^{\frac{\alpha}{2}\left\langle
v\right\rangle ^{\gamma}}\mathrm{P}_{1}\widehat{f}\left(  t,\eta,v\right)
\right\vert _{L_{\gamma-1}^{2}}^{2}\leq C_{\sigma}\left\vert \eta\right\vert
^{2}\left\vert \widehat{f}\right\vert _{L_{\gamma-1}^{2}}^{2}+C_{\gamma
}\left\vert \mathrm{P}_{1}\widehat{f}\right\vert _{L_{v}^{2}\left(
B_{2R}\right)  }^{2}, \label{Micro-weighted ineq}%
\end{equation}
where the constants $C_{\gamma}>0$ and $R>0$ are dependent only upon $\gamma.$
We split the solution $f$ into two parts: $f=P_{0}f+\mathrm{P}_{1}f,$ and then
apply $\mathrm{P}_{1}$ to equation $\left(  \ref{FT}\right)  $:%
\[
\partial_{t}\mathrm{P}_{1}\widehat{f}+iv\cdot\eta\mathrm{P}_{1}\widehat
{f}-L\mathrm{P}_{1}\widehat{f}=-\mathrm{P}_{1}\left(  iv\cdot\eta
P_{0}\widehat{f}\right)  +P_{0}\left(  iv\cdot\eta\mathrm{P}_{1}\widehat
{f}\right)  .
\]
Multiply the above equation by $e^{\alpha\left\langle v\right\rangle ^{\gamma
}}\overline{\mathrm{P}_{1}\widehat{f}}$ and integrate with respect to $v$ to
obtain%
\[
\frac{1}{2}\frac{d}{dt}\left\vert e^{\frac{\alpha}{2}\left\langle
v\right\rangle ^{\gamma}}\mathrm{P}_{1}\widehat{f}\left(  t,\eta,v\right)
\right\vert _{L_{v}^{2}}^{2}-\operatorname{Re}\left\langle e^{\alpha
\left\langle v\right\rangle ^{\gamma}}L\mathrm{P}_{1}\widehat{f}%
,\mathrm{P}_{1}\widehat{f}\right\rangle _{v}=\Gamma,
\]
where
\[
\Gamma=-\operatorname{Re}\left\langle \mathrm{P}_{1}\left(  iv\cdot\eta
P_{0}\widehat{f}\right)  ,e^{\alpha\left\langle v\right\rangle ^{\gamma}%
}\mathrm{P}_{1}\widehat{f}\right\rangle +\operatorname{Re}\left\langle
P_{0}\left(  iv\cdot\eta\mathrm{P}_{1}\widehat{f}\right)  ,e^{\alpha
\left\langle v\right\rangle ^{\gamma}}\mathrm{P}_{1}\widehat{f}\right\rangle
.
\]
Owing to the rapid decay of $\mathcal{M}^{1/2},$ we obtain%
\[
\left\vert \Gamma\right\vert \leq\epsilon\left\vert e^{\frac{\alpha}%
{2}\left\langle v\right\rangle ^{\gamma}}\mathrm{P}_{1}\widehat{f}\left(
t,\eta,v\right)  \right\vert _{L_{\gamma-1}^{2}}^{2}+C_{\epsilon}\left\vert
\eta\right\vert ^{2}\left(  \left\vert \mathrm{P}_{1}\widehat{f}\left(
t,\eta,v\right)  \right\vert _{L_{\gamma-1}^{2}}^{2}+\left\vert P_{0}%
\widehat{f}\right\vert _{L_{v}^{2}}^{2}\right)  ,
\]
which holds for any small $\epsilon>0.$ On the other hand, we rewrite
$L=-\Lambda+K,$ $K=\varpi\chi_{R}\left(  \left\vert v\right\vert \right)  ,$
where $R>0$ and $\varpi>0\ $are chosen sufficiently large such that%
\[
\frac{\left\vert v\right\vert ^{2}\left\langle v\right\rangle ^{2\gamma-4}}%
{4}-\frac{3}{2}\left\langle v\right\rangle ^{\gamma-2}-\frac{\left(
\gamma-2\right)  }{2}\left\vert v\right\vert ^{2}\left\langle v\right\rangle
^{\gamma-4}+\varpi\chi_{R}\left(  \left\vert v\right\vert \right)  \geq
\frac{1}{5}\left\langle v\right\rangle ^{2\gamma-2}.
\]
Hence, we have%
\begin{align*}
&  \quad-\operatorname{Re}\left\langle e^{\alpha\left\langle v\right\rangle
^{\gamma}}L\mathrm{P}_{1}\widehat{f},\mathrm{P}_{1}\widehat{f}\right\rangle
_{v}\\
&  =\operatorname{Re}\int e^{\alpha\left\langle v\right\rangle ^{\gamma}%
}\left[  \left(  \Lambda-K\right)  \mathrm{P}_{1}\widehat{f}\right]
\mathrm{P}_{1}\overline{\widehat{f}}dv\\
&  \geq\int e^{\alpha\left\langle v\right\rangle ^{\gamma}}\left\vert
\nabla_{v}\mathrm{P}_{1}\widehat{f}\right\vert ^{2}+\operatorname{Re}%
\int\alpha\gamma\left\langle v\right\rangle ^{\gamma-2}e^{\alpha\left\langle
v\right\rangle ^{\gamma}}\left(  v\cdot\nabla_{v}\mathrm{P}_{1}\widehat
{f}\right)  \mathrm{P}_{1}\overline{\widehat{f}}dv\\
&  \quad+\frac{1}{5}\int\left\langle v\right\rangle ^{2\gamma-2}%
e^{\alpha\left\langle v\right\rangle ^{\gamma}}\left\vert \mathrm{P}%
_{1}\widehat{f}\right\vert ^{2}dv-C^{\prime}\left\vert \mathrm{P}_{1}%
\widehat{f}\right\vert _{L_{v}^{2}\left(  B_{2R}\right)  }^{2},
\end{align*}
where $C^{\prime}=C^{\prime}\left(  \alpha,\gamma,R\right)  .$ Note that
$\alpha\gamma<1/20,$ the Cauchy-Schwartz inequality implies
\begin{align*}
&  \quad\left\vert \operatorname{Re}\int\alpha\gamma\left\langle
v\right\rangle ^{\gamma-2}e^{\alpha\left\langle v\right\rangle ^{\gamma}%
}\left(  v\cdot\nabla_{v}\mathrm{P}_{1}\widehat{f}\right)  \mathrm{P}%
_{1}\overline{\widehat{f}}dv\right\vert \\
&  \leq\int\alpha\gamma\left\langle v\right\rangle ^{\gamma-1}e^{\alpha
\left\langle v\right\rangle ^{\gamma}}\left\vert \nabla_{v}\mathrm{P}%
_{1}\widehat{f}\right\vert \left\vert \mathrm{P}_{1}\widehat{f}\right\vert
dv\\
&  \leq\int e^{\alpha\left\langle v\right\rangle ^{\gamma}}\left\vert
\nabla_{v}\mathrm{P}_{1}\widehat{f}\right\vert ^{2}dv+\frac{1}{80}%
\int\left\langle v\right\rangle ^{2\gamma-2}e^{\alpha\left\langle
v\right\rangle ^{\gamma}}\left\vert \mathrm{P}_{1}\widehat{f}\right\vert
^{2}dv,
\end{align*}
so we deduce
\[
-\operatorname{Re}\left\langle \left\langle v\right\rangle ^{2\ell}%
L\mathrm{P}_{1}\widehat{f},\mathrm{P}_{1}\widehat{f}\right\rangle _{v}%
\geq\frac{1}{6}\int\left\langle v\right\rangle ^{2\gamma-2}e^{\alpha
\left\langle v\right\rangle ^{\gamma}}\left\vert \mathrm{P}_{1}\widehat
{f}\right\vert ^{2}dv-\widetilde{C}\left(  R,\gamma,\alpha\right)  \left\vert
\mathrm{P}_{1}\widehat{f}\right\vert _{L_{v}^{2}\left(  B_{2R}\right)  }^{2},
\]
where $\widetilde{C}\left(  R,\ell,\alpha\right)  >0.\ $Consequently,
\[
\frac{d}{dt}\left\vert e^{\frac{\alpha}{2}\left\langle v\right\rangle
^{\gamma}}\mathrm{P}_{1}\widehat{f}\left(  t,\eta,v\right)  \right\vert
_{L_{v}^{2}}^{2}+\sigma\left\vert e^{\frac{\alpha}{2}\left\langle
v\right\rangle ^{\gamma}}\mathrm{P}_{1}\widehat{f}\left(  t,\eta,v\right)
\right\vert _{L_{\gamma-1}^{2}}^{2}\leq C_{\sigma}\left\vert \eta\right\vert
^{2}\left\vert \widehat{f}\right\vert _{L_{\gamma-1}^{2}}^{2}+C_{\gamma
}\left\vert \mathrm{P}_{1}\widehat{f}\right\vert _{L_{v}^{2}\left(
B_{2R}\right)  }^{2},
\]
for some constant $\sigma>0.$ \bigskip In addition, if we multiply $\left(
\ref{FT}\right)  $ with $e^{\alpha\left\langle v\right\rangle ^{\gamma}%
}\overline{\widehat{f}\left(  t,\eta,v\right)  },$ integrate in $v$ and use
the same procedure as above, we also obtain%
\begin{equation}
\frac{1}{2}\frac{d}{dt}\left\vert e^{\frac{\alpha}{2}\left\langle
v\right\rangle ^{\gamma}}\widehat{f}\left(  t,\eta,v\right)  \right\vert
_{L_{v}^{2}}^{2}+\sigma\left\vert e^{\frac{\alpha}{2}\left\langle
v\right\rangle ^{\gamma}}\widehat{f}\left(  t,\eta,v\right)  \right\vert
_{L_{\gamma-1}^{2}}^{2}\leq C_{\gamma}\left\vert \widehat{f}\right\vert
_{L^{2}\left(  B_{2R}\right)  }^{2}. \label{weighted ineq}%
\end{equation}
{\ }

To do the weighted estimate, we introduce a new energy as follows :
\[
\widetilde{\mathcal{E}}\left(  t,\eta\right)  :=\widetilde{\mathcal{E}}%
^{0}\left(  t,\eta\right)  +\widetilde{\mathcal{E}}^{1}\left(  t,\eta\right)
,
\]
with%
\[
\widetilde{\mathcal{E}}^{0}\left(  t,\eta\right)  =1_{\left\vert
\eta\right\vert \leq1}\left(  \mathcal{E}\left(  t,\eta\right)  +\kappa
_{4}\left\vert e^{\frac{\alpha}{2}\left\langle v\right\rangle ^{\gamma}%
}\mathrm{P}_{1}\widehat{f}\left(  t,\eta,v\right)  \right\vert _{L_{v}^{2}%
}^{2}\right)  ,
\]%
\[
\widetilde{\mathcal{E}}^{1}\left(  t,\eta\right)  =1_{\left\vert
\eta\right\vert >1}\left(  \mathcal{E}\left(  t,\eta\right)  +\kappa
_{5}\left\vert e^{\frac{\alpha}{2}\left\langle v\right\rangle ^{\gamma}%
}\widehat{f}\left(  t,\eta,v\right)  \right\vert _{L_{v}^{2}}^{2}\right)  ,
\]
where $\mathcal{E}\left(  t,\eta\right)  $ is defined as in $\left(
\ref{Lyapu 1}\right)  $ and the constants $\kappa_{4},$ $\kappa_{5}>0$ will be
chosen small enough. Notice further that $\left\vert \widehat{a}\right\vert
^{2}=\left\vert P_{0}\widehat{f}\right\vert _{L_{v}^{2}}^{2}\gtrsim\left\vert
e^{\frac{\alpha}{2}\left\langle v\right\rangle ^{\gamma}}P_{0}\widehat
{f}\right\vert _{L_{v}^{2}}^{2}$ for all $0<\alpha\gamma<1/20,$ and so
$\widetilde{\mathcal{E}}\left(  t,\eta\right)  \approx\left\vert
e^{\frac{\alpha}{2}\left\langle v\right\rangle ^{\gamma}}\widehat
{f}\right\vert _{L_{v}^{2}}^{2}.$

For $\widetilde{\mathcal{E}}^{1}\left(  t,\eta\right)  ,$ we combine $\left(
\ref{E}\right)  $ and $\left(  \ref{weighted ineq}\right)  $ for $\left\vert
\eta\right\vert >1$ to obtain
\begin{equation}
\partial_{t}\widetilde{\mathcal{E}}^{1}\left(  t,\eta\right)  +\sigma
\left\vert e^{\frac{\alpha}{2}\left\langle v\right\rangle ^{\gamma}}%
\widehat{f}\left(  t,\eta,v\right)  \right\vert _{L_{\gamma-1}^{2}}%
^{2}1_{\left\vert \eta\right\vert >1}\leq0, \label{E0}%
\end{equation}
for $\kappa_{5}>0$ small enough, since $\left\vert \eta\right\vert
^{2}/\left(  1+\left\vert \eta\right\vert ^{2}\right)  \geq\frac{1}{2}.$

For $\widetilde{\mathcal{E}}^{0}\left(  t,\eta\right)  ,$ since $\left\vert
\eta\right\vert ^{2}/\left(  1+\left\vert \eta\right\vert ^{2}\right)
\geq\frac{\left\vert \eta\right\vert ^{2}}{2}$ for $\left\vert \eta\right\vert
\leq1$ and $\left\vert \widehat{a}\right\vert ^{2}\gtrsim\left\vert
e^{\frac{\alpha}{2}\left\langle v\right\rangle ^{\gamma}}P_{0}\widehat
{f}\right\vert _{L_{v}^{2}}^{2}$ for all $\alpha\gamma<1/20,$ combining
$\left(  \ref{E}\right)  $ and $\left(  \ref{Micro-weighted ineq}\right)  $
for $\left\vert \eta\right\vert \leq1$ gives
\begin{equation}
\partial_{t}\widetilde{\mathcal{E}}^{0}\left(  t,\eta\right)  +\sigma
\left\vert \eta\right\vert ^{2}\left\vert e^{\frac{\alpha}{2}\left\langle
v\right\rangle ^{\gamma}}\widehat{f}\left(  t,\eta,v\right)  \right\vert
_{L_{\gamma-1}^{2}}^{2}1_{\left\vert \eta\right\vert \leq1}\leq0, \label{E1}%
\end{equation}
for $\kappa_{4}>0$ small enough. This completes the proof.
\end{proof}

Now, it is enough to prove the estimate in the time-like region. We apply the
H\"{o}lder inequality to obtain that for $j\geq1,${
\begin{align*}
\mathcal{E}\left(  t,\eta\right)   &  \lesssim\left\vert \widehat{f}\left(
t,\eta,v\right)  \right\vert _{L_{v}^{2}}^{2}=\int\left(  \left\vert
\widehat{f}(t,\eta,v)\right\vert ^{2}e^{\alpha\left\langle v\right\rangle
^{\gamma}}\right)  ^{\frac{1}{j+1}}\left(  \left\vert \widehat{f}%
(t,\eta,v)\right\vert ^{2}e^{-\frac{\alpha}{j}\left\langle v\right\rangle
^{\gamma}}\right)  ^{\frac{j}{j+1}}dv\\
&  \leq\left(  \int e^{-\frac{\alpha}{j}\left\langle v\right\rangle ^{\gamma}%
}\left\vert \widehat{f}\left(  t,\eta,v\right)  \right\vert ^{2}dv\right)
^{j/\left(  j+1\right)  }\left(  \int e^{\alpha\left\langle v\right\rangle
^{\gamma}}\left\vert \widehat{f}\left(  t,\eta,v\right)  \right\vert
^{2}dv\right)  ^{1/\left(  j+1\right)  }\\
&  \lesssim\left\vert \widehat{f}\left(  t,\eta,v\right)  \right\vert
_{L_{\gamma-1}^{2}}^{2j/\left(  j+1\right)  }\widetilde{\mathcal{E}%
}^{1/\left(  j+1\right)  }\left(  t,\eta\right)  \,.
\end{align*}
} Thus we conclude that
\[
\mathcal{E}^{\left(  j+1\right)  /j}\left(  t,\eta\right)  \lesssim\left\vert
\widehat{f}\left(  t,\eta,v\right)  \right\vert _{L_{\gamma-1}^{2}}%
^{2}\widetilde{\mathcal{E}}^{1/j}\left(  t,\eta\right)  \lesssim\left\vert
\widehat{f}\left(  t,\eta,v\right)  \right\vert _{L_{\gamma-1}^{2}}%
^{2}\widetilde{\mathcal{E}}^{1/j}\left(  0,\eta\right)  .
\]
Now we can rewrite $\left(  \ref{E}\right)  ,$ for any $\eta\in\mathbb{R}%
^{3},$ as
\[
\partial_{t}\mathcal{E}\left(  t,\eta\right)  +\sigma\widehat{\rho}\left(
\eta\right)  \mathcal{E}^{\left(  j+1\right)  /j}\left(  t,\eta\right)
\widetilde{\mathcal{E}}^{-1/j}\left(  0,\eta\right)  \leq0.
\]
Integrating this over time, we obtain%
\[
j\mathcal{E}^{-1/j}\left(  0,\eta\right)  -j\mathcal{E}^{-1/j}\left(
t,\eta\right)  \lesssim-t\widehat{\rho}\left(  \eta\right)  \widetilde
{\mathcal{E}}^{-1/j}\left(  0,\eta\right)  .
\]
As a consequence, for any $j\geq1,$ uniformly in $\eta\in\mathbb{R}^{3},$ we
get
\[
\mathcal{E}\left(  t,\eta\right)  \lesssim\widetilde{\mathcal{E}}\left(
0,\eta\right)  \left(  \frac{t\widehat{\rho}\left(  \eta\right)  }%
{j}+1\right)  ^{-j}.
\]

Recall that the long wave part $f_{L}$ and the short wave part $f_{S}$ of the
solution $f$ are given respectively by
\[
f_{L}=\int_{\left\vert \eta\right\vert \leq1}e^{i\eta\cdot x+\left(
-iv\cdot\eta+L\right)  t}\widehat{f_{0}}\left(  \eta,v\right)  d\eta,
\]%
\[
f_{S}=\int_{\left\vert \eta\right\vert >1}e^{i\eta\cdot x+\left(  -iv\cdot
\eta+L\right)  t}\widehat{f_{0}}\left(  \eta,v\right)  d\eta,
\]
When $\left\vert \eta\right\vert \leq1$ and $k\in{\mathbb{N\cup\{}}0\}$,
since
\[
\int_{\left\vert \eta\right\vert \leq1}|\eta|^{2k}\left(  \frac{t\left\vert
\eta\right\vert ^{2}}{j}+1\right)  ^{-j}d\eta\lesssim(1+t)^{-\frac{3}{2}%
-k}\text{ if }j>\frac{3}{2}+k,
\]
we obtain
\begin{align}
\int_{\left\vert \eta\right\vert \leq1}|\eta|^{2k}\mathcal{E}\left(
t,\eta\right)  d\eta &  \lesssim\int_{\left\vert \eta\right\vert \leq1}%
|\eta|^{2k}\left(  \frac{t\left\vert \eta\right\vert ^{2}}{j}+1\right)
^{-j}\widetilde{\mathcal{E}}\left(  0,\eta\right)  d\eta\label{Elong}\\
&  \lesssim(1+t)^{-\frac{3}{2}-k}\left\Vert e^{\frac{\alpha}{2}\left\langle
v\right\rangle ^{\gamma}}f_{0}\right\Vert _{L_{x}^{1}L_{v}^{2}}^{2},\nonumber
\end{align}
which implies that
\[
\left\Vert \nabla_{x}^{k}f_{L}\right\Vert _{L^{2}}\lesssim(1+t)^{-\frac{3}%
{4}-\frac{k}{2}}\left\Vert e^{\frac{\alpha}{2}\left\langle v\right\rangle
^{\gamma}}f_{0}\right\Vert _{L_{x}^{1}L_{v}^{2}}.
\]
By the Sobolev inequality, we get
\[
\left\Vert f_{L}\right\Vert _{L_{x}^{\infty}L_{v}^{2}}\lesssim\left\Vert
\nabla_{x}^{2}f_{L}\right\Vert _{L^{2}}^{3/4}\left\Vert f_{L}\right\Vert
_{L^{2}}^{1/4}\lesssim(1+t)^{-\frac{3}{2}}\left\Vert f_{0}\right\Vert
_{L_{x}^{1}L_{v}^{2}\left(  e^{\alpha\left\langle v\right\rangle ^{\gamma}%
}\right)  }.
\]

{When $\left\vert \eta\right\vert >1,$ we note that the equations \ref{E-N}
and \ref{EL} for $f_{S}$ are similar to \ref{3} and \ref{4} for $h^{(0)}$.
Then following the similar procedure of the proof,} it implies
\[
\left\Vert f_{S}\right\Vert _{L^{2}}\lesssim e^{-c_{\gamma}\alpha
^{\frac{2\left(  1-\gamma\right)  }{2-\gamma}}t^{\frac{\gamma}{2-\gamma}}%
}\left\Vert f_{0}\right\Vert _{L^{2}(e^{\alpha\left\langle v\right\rangle
^{\gamma}})}\,,\ \ t\geq0,
\]
for some constant $c_{\gamma}>0.$

To sum up, we have the following proposition:

\begin{proposition}
\label{LS-estimate2}Let $0<\gamma<1\ $and let $f$ be the solution of equation
$\left(  \ref{in.1.c}\right)  .$ For any $\alpha>0$ small with $\alpha
\gamma<1/20$, we have \newline\noindent\textrm{$(i)$ (Long wave $f_{L}$)}
\begin{equation}
\left\Vert f_{L}\right\Vert _{L_{x}^{\infty}L_{v}^{2}}\lesssim(1+t)^{-\frac
{3}{2}}\left\Vert f_{0}\right\Vert _{L_{x}^{1}L_{v}^{2}\left(  e^{\alpha
\left\langle v\right\rangle ^{\gamma}}\right)  }. \label{f_S-long}%
\end{equation}

\noindent\textrm{$(ii)$ (Short wave $f_{S}$)} There exists $c_{\gamma}>0$ such
that
\begin{equation}
\left\Vert f_{S}\right\Vert _{L^{2}}\lesssim e^{-c_{\gamma}\alpha
^{\frac{2\left(  1-\gamma\right)  }{2-\gamma}}t^{\frac{\gamma}{2-\gamma}}%
}\left\Vert f_{0}\right\Vert _{L^{2}(e^{\alpha\left\langle v\right\rangle
^{\gamma}})}\,. \label{f_S-soft}%
\end{equation}

\end{proposition}

Based on the long wave-short wave decomposition and wave -remainder
decomposition, i.e.,
\[
f=f_{L}+f_{S}=W^{\left(  3\right)  }+\mathcal{R}^{(3)},
\]
we now define the tail part as $f_{R}=\mathcal{R}^{(3)}-f_{L}=f_{S}-W^{\left(
3\right)  },$ which leads to that $f$ can be written as $f=W^{\left(
3\right)  }+f_{L}+f_{R}.$ From Lemma \ref{second-der},
\begin{equation}
\left\Vert \mathcal{R}^{(3)}(t)\right\Vert _{H_{x}^{2}L_{v}^{2}}\lesssim
\int_{0}^{t}\left\Vert h^{(3)}(s)\right\Vert _{H_{x}^{2}L_{v}^{2}}%
ds\lesssim\left(  1+t^{4}\right)  \left\Vert f_{0}\right\Vert _{L^{2}},
\label{Prop19-1}%
\end{equation}
and so
\[
\left\Vert f_{R}\right\Vert _{H_{x}^{2}L_{v}^{2}}=\left\Vert \mathcal{R}%
^{(3)}-f_{L}\right\Vert _{H_{x}^{2}L_{v}^{2}}\lesssim\left(  1+t^{4}\right)
\left\Vert f_{0}\right\Vert _{L^{2}},\ \ t>0.
\]
In view of Proposition \ref{pointwise-wave part <1} and Proposition
\ref{LS-estimate2},
\[
\left\Vert f_{R}\right\Vert _{L^{2}}=\left\Vert f_{S}-W^{\left(  3\right)
}\right\Vert _{L^{2}}\lesssim e^{-\frac{c_{\gamma}}{2}\alpha^{\frac{2\left(
1-\gamma\right)  }{2-\gamma}}t^{\frac{\gamma}{2-\gamma}}}\left\Vert
f_{0}\right\Vert _{L^{2}(e^{4\alpha\left\langle v\right\rangle ^{\gamma}}%
)},\ \ t>0.\
\]
The Sobolev inequality implies
\begin{equation}
\left\vert f_{R}\right\vert _{L_{v}^{2}}\leq\left\Vert f_{R}\right\Vert
_{L_{v}^{2}L_{x}^{\infty}}\lesssim\left\Vert f_{R}\right\Vert _{H_{x}^{2}%
L_{v}^{2}}^{3/4}\left\Vert f_{R}\right\Vert _{L^{2}}^{1/4}\lesssim
e^{-\frac{c_{\gamma}}{16}\alpha^{\frac{2\left(  1-\gamma\right)  }{2-\gamma}%
}t^{\frac{\gamma}{2-\gamma}}}\left\Vert f_{0}\right\Vert _{L^{2}%
(e^{4\alpha\left\langle v\right\rangle ^{\gamma}})},\ \ t>0.
\label{fR-gamma<1}%
\end{equation}

Combining (\ref{Wave-ptw-gamma<1}), (\ref{f_S-long}) and (\ref{fR-gamma<1}),
we obtain the pointwise estimate for the solution in the time-like region.

\begin{theorem}
[Time-like region for $0<\gamma<1$]%
\label{time-like region for gamma less than 1} Let $0<\gamma<1$ and let $f$ be
the solution to equation (\ref{in.1.c}). Assume that the initial condition
$f_{0}$ has compact support in the $x$ variable and is bounded in $L_{v}%
^{2}(e^{4\alpha\left\langle v\right\rangle ^{\gamma}})$. Then for $\alpha>0$
is small enough, there exists a positive constant $c_{\gamma}$ such that
\begin{equation}
\left\vert f\right\vert _{L_{v}^{2}}\lesssim\left[  (1+t)^{-3/2}%
+(1+t^{-9/4})e^{-c_{\gamma}\alpha^{\frac{2\left(  1-\gamma\right)  }{2-\gamma
}}t^{\frac{\gamma}{2-\gamma}}}\right]  \Vert f_{0}\Vert_{L_{x}^{\infty}%
L_{v}^{2}(e^{4\alpha\left\langle v\right\rangle ^{\gamma}})}\,.
\end{equation}

\end{theorem}

\section{In the space-like region}

\label{layer} We have finished the estimate of solution inside the time-like
region. To have the global picture of the space-time structure of solution, we
still need to investigate the solution in the space-like region. To this end,
we shall estimate the wave part $W^{(3)}$ and the remainder part
$\mathcal{R}^{(3)}$ separately. Here, the weighted energy estimate plays a
decisive role.


\subsection{The case $\gamma\geq3/2$: Exponential decay}

\begin{proposition}
\label{weig_1} Consider the weight functions
\[
w(x,t)=e^{\frac{\left\langle x\right\rangle -Mt}{2D}}\,,\quad\mu
(x)=e^{\frac{\left\langle x\right\rangle }{D}}\,,
\]
where $D$ and $M$ are chosen sufficiently large. Then for $0\leq j\leq3$, we
have%
\begin{equation}
\Vert wh^{(j)}\Vert_{H_{x}^{2}L_{v}^{2}}\lesssim t^{-3+j}\Vert f_{0}%
\Vert_{L^{2}(\mu)},\ \ \ \ \ 0<t\leq1, \label{energy 1}%
\end{equation}%
\begin{equation}
\Vert wh^{(j)}\Vert_{H_{x}^{2}L_{v}^{2}}\lesssim e^{-C t}\Vert f_{0}%
\Vert_{L^{2}(\mu)},\ \ \ \ \ t>1, \label{energy 2}%
\end{equation}
and
\begin{equation}
\Vert w\mathcal{R}^{(3)}\Vert_{H_{x}^{2}L_{v}^{2}}\lesssim\Vert f_{0}%
\Vert_{L^{2}(\mu)},\ \ \ \ \ \ t>0. \label{energy 3}%
\end{equation}

\end{proposition}

\begin{proof}
In view of that $w(x,t)$ is non-increasing in $t$, it is not hard to verify
that
\[
\left\Vert wg(t)\right\Vert _{H_{x}^{2}L_{v}^{2}}\lesssim\left\Vert
g(t)\right\Vert _{H_{x}^{2}L_{v}^{2}(\mu)}\,.
\]
Then the weighted energy inequalities (\ref{energy 1}) and (\ref{energy 2})
follow from Lemma \ref{second-der} directly.

It remains to show the weighted energy estimates for the remainder part
$\mathcal{R}^{(3)}$, $t>0$. We shall demonstrate that
\[
\Vert w\mathcal{R}^{(3)}\Vert_{H_{x}^{2}L_{v}^{2}}\lesssim\left(  1+t\right)
\Vert f_{0}\Vert_{L^{2}(\mu)},\ \ \ \ \ \ t>0.
\]
To see this, let $u=w\mathcal{R}^{(3)}$and then $\partial_{x}^{\beta}u,$ where
$\beta$ is a multi-index, solves the equation
\begin{align*}
\partial_{t}\left(  \partial_{x}^{\beta}u\right)   &  =-v\cdot\nabla
_{x}\left(  \partial_{x}^{\beta}u\right)  -\frac{1}{2D}\left(  M-\frac{x\cdot
v}{\left\langle x\right\rangle }\right)  \partial_{x}^{\beta}u+L\partial
_{x}^{\beta}u+K\partial_{x}^{\beta}\left(  wh^{(3)}\right)  \,\\
&  \quad+\frac{1}{2D}\sum_{\substack{\beta_{1}+\beta_{2}=\beta\\\left\vert
\beta_{1}\right\vert \geq1}}\binom{\beta}{\beta_{1}\,\beta_{2}}\partial
_{x}^{\beta_{1}}\left(  \frac{x}{\left\langle x\right\rangle }\right)  \cdot
v\partial_{x}^{\beta_{2}}u.
\end{align*}
The energy estimate gives
\begin{align*}
\frac{1}{2}\partial_{t}\left\Vert \partial_{x}^{\beta}u\right\Vert _{L^{2}%
}^{2}  &  =-\frac{1}{2D}\int\left(  M-\frac{x\cdot v}{\left\langle
x\right\rangle }\right)  \left\vert \partial_{x}^{\beta}u\right\vert
^{2}dxdv+\int\left(  L\partial_{x}^{\beta}u\right)  \partial_{x}^{\beta
}udxdv\\
&  \quad+\frac{1}{2D}\int\sum_{\substack{\beta_{1}+\beta_{2}=\beta\\\left\vert
\beta_{1}\right\vert \geq1}}\binom{\beta}{\beta_{1}\,\beta_{2}}\partial
_{x}^{\beta_{1}}\left(  \frac{x}{\left\langle x\right\rangle }\right)  \cdot
v\partial_{x}^{\beta_{2}}u\partial_{x}^{\beta}udxdv\\
&  \quad+\int\partial_{x}^{\beta}uK\partial_{x}^{\beta}\left(  wh^{(3)}%
\right)  dxdv.
\end{align*}
Note that $2\gamma-2\geq1$ if $\gamma\geq3/2,$ and recall that $\Lambda=-L+K,$
hence%
\begin{align*}
\left\vert \int\frac{x\cdot v}{\left\langle x\right\rangle }\left\vert
\partial_{x}^{\beta}u\right\vert ^{2}dxdv\right\vert  &  \leq\int\left\langle
v\right\rangle ^{2\gamma-2}\left\vert \partial_{x}^{\beta}u\right\vert
^{2}dxdv\lesssim\int\left(  \Lambda\partial_{x}^{\beta}u\right)  \partial
_{x}^{\beta}udxdv\\
&  \lesssim-\int\left(  L\partial_{x}^{\beta}u\right)  \partial_{x}^{\beta
}udxdv+\int\left\vert \partial_{x}^{\beta}u\right\vert ^{2}dxdv,
\end{align*}
and
\begin{align*}
&  \quad\left\vert \int\sum_{\substack{\beta_{1}+\beta_{2}=\beta\\\left\vert
\beta_{1}\right\vert \geq1}}\binom{\beta}{\beta_{1}\,\beta_{2}}\partial
_{x}^{\beta_{1}}\left(  \frac{x}{\left\langle x\right\rangle }\right)  \cdot
v\partial_{x}^{\beta_{2}}u\partial_{x}^{\beta}udxdv\right\vert \\
&  \lesssim\int\sum_{\substack{\beta_{1}+\beta_{2}=\beta\\\left\vert \beta
_{1}\right\vert \geq1}}\left\langle v\right\rangle ^{2\gamma-2}\left\vert
\partial_{x}^{\beta_{2}}u\partial_{x}^{\beta}u\right\vert dxdv\\
&  \lesssim\int\sum_{\substack{\beta_{1}+\beta_{2}=\beta\\\left\vert \beta
_{1}\right\vert \geq1}}\left(  \left(  -L\partial_{x}^{\beta_{2}}u\right)
\partial_{x}^{\beta_{2}}u+\left(  -L\partial_{x}^{\beta}u\right)  \partial
_{x}^{\beta}u\right)  dxdv+\int\sum_{\substack{\beta_{1}+\beta_{2}%
=\beta\\\left\vert \beta_{1}\right\vert \geq1}}\left(  \left\vert \partial
_{x}^{\beta_{2}}u\right\vert ^{2}+\left\vert \partial_{x}^{\beta}u\right\vert
^{2}\right)  dxdv.
\end{align*}
Also,
\[
\left\vert \int\partial_{x}^{\beta}uK\partial_{x}^{\beta}\left(
wh^{(3)}\right)  dxdv\right\vert \lesssim\Vert\partial_{x}^{\beta}%
u\Vert_{_{L^{2}}}\Vert\partial_{x}^{\beta}\left(  wh^{(3)}\right)
\Vert_{_{L^{2}}}.
\]
After choosing $D$ and $M$ large enough, we have
\[
\frac{d}{dt}\Vert u\Vert_{H_{x}^{2}L_{v}^{2}}\lesssim\Vert wh^{(3)}%
\Vert_{_{H_{x}^{2}L_{v}^{2}}}\,.
\]
Hence, it follows from (\ref{energy 1}) and (\ref{energy 2}) that
\[
\Vert u\Vert_{H_{x}^{2}L_{v}^{2}}\left(  t\right)  \lesssim\int_{0}^{t}\Vert
wh^{(3)}\Vert_{_{H_{x}^{2}L_{v}^{2}}}\,\left(  s\right)  ds\lesssim\Vert
f_{0}\Vert_{L^{2}(\mu)}\,.
\]

\end{proof}

Note that $w(x,t)\geq e^{\frac{\left\langle x\right\rangle +2Mt}{8D}}$ if
$\left\langle x\right\rangle \geq2Mt$, hence for $\gamma\geq3/2$, the Sobolev
inequality implies%
\begin{align*}
e^{\frac{\left\langle x\right\rangle +2Mt}{8D}}|f|_{L_{v}^{2}}  &  \leq
\sum_{j=0}^{3}\left\vert wh^{\left(  j\right)  }\right\vert _{L_{v}^{2}%
}+\left\vert w\mathcal{R}^{(3)}\right\vert _{L_{v}^{2}}\\
&  \lesssim\sum_{j=0}^{3}\left\Vert wh^{\left(  j\right)  }\right\Vert
_{H_{x}^{2}L_{v}^{2}}^{3/4}\left\Vert wh^{\left(  j\right)  }\right\Vert
_{L^{2}}^{1/4}+\left\Vert w\mathcal{R}^{\left(  3\right)  }\right\Vert
_{H_{x}^{2}L_{v}^{2}}^{3/4}\left\Vert w\mathcal{R}^{\left(  3\right)
}\right\Vert _{L^{2}}^{1/4}\\
&  \lesssim\left(  t^{-9/4}+1\right)  \Vert f_{0}\Vert_{L^{2}(\mu)}\\
&  \lesssim\left(  t^{-9/4}+1\right)  \Vert f_{0}\Vert_{L^{2}}\,.
\end{align*}
The last inequality is due to the compact support assumption of the initial data.

\begin{theorem}
[Space-like region for $\gamma\geq3/2$]%
\label{space-like region for gamma greater or equal to 3/2} Let $\gamma
\geq3/2$ and let $f$ be the solution to equation (\ref{in.1.c}). Assume that
the initial condition $f_{0}$ has compact support in the $x$ variable and is
bounded in $L_{v}^{2}$. Then there exists a large positive constant $M$ such
that if $\left\langle x\right\rangle >2Mt$, we have
\[
\left\vert f\right\vert _{L_{v}^{2}}\lesssim(1+t^{-9/4})e^{-C\left(
\left\langle x\right\rangle +t\right)  }\Vert f_{0}\Vert_{L^{2}}\,,
\]
here $C=C(M)$ is a positive constant.
\end{theorem}

\subsection{The case $0<\gamma<3/2$: Subexponential decay}

If $0<\gamma<3/2$, we consider the weight functions
\[
w(t,x,v)=e^{\frac{\alpha\rho(t,x,v)}{2}}\,,\quad\mu(x,v)=e^{\alpha c(x,v)}\,,
\]
where
\begin{align*}
\rho(t,x,v)  &  =5\left(  \delta(\left\langle x\right\rangle -Mt)\right)
^{\frac{\gamma}{3-\gamma}}\left(  1-\chi\left(  \delta\left(  \left\langle
x\right\rangle -Mt\right)  \left\langle v\right\rangle ^{\gamma-3}\right)
\right) \\
&  \quad+\left[  \left(  1-\chi\left(  \delta\left(  \left\langle
x\right\rangle -Mt\right)  \left\langle v\right\rangle ^{\gamma-3}\right)
\right)  \delta(\left\langle x\right\rangle -Mt)\left\langle v\right\rangle
^{2\gamma-3}+3\left\langle v\right\rangle ^{\gamma}\right]  \chi\left(
\delta\left(  \left\langle x\right\rangle -Mt\right)  \left\langle
v\right\rangle ^{\gamma-3}\right)  \,,
\end{align*}
and
\begin{align*}
c(x,v)  &  =5\left(  \delta\left\langle x\right\rangle \right)  ^{\frac
{\gamma}{3-\gamma}}\left(  1-\chi\left(  \delta\left\langle x\right\rangle
\left\langle v\right\rangle ^{\gamma-3}\right)  \right) \\
&  \quad+\left[  \left(  1-\chi\left(  \delta\left\langle x\right\rangle
\left\langle v\right\rangle ^{\gamma-3}\right)  \right)  \delta\left\langle
x\right\rangle \left\langle v\right\rangle ^{2\gamma-3}+3\left\langle
v\right\rangle ^{\gamma}\right]  \chi\left(  \delta\left\langle x\right\rangle
\left\langle v\right\rangle ^{\gamma-3}\right)  \,.
\end{align*}
Here $M$ is a large positive constant, $\delta,$ $\alpha$ are small positive
constants; all of them will be chosen later. We introduce the following
space-velocity decomposition:
\[
H_{+}=\{(x,v):[\delta(\left\langle x\right\rangle -Mt)]\geq2\left\langle
v\right\rangle ^{3-\gamma}\}\,,
\]%
\[
H_{0}=\{(x,v):\left\langle v\right\rangle ^{3-\gamma}<[\delta(\left\langle
x\right\rangle -Mt)]<2\left\langle v\right\rangle ^{3-\gamma}\}\,,
\]
and
\[
H_{-}=\{(x,v):[\delta(\left\langle x\right\rangle -Mt)]\leq\left\langle
v\right\rangle ^{3-\gamma}\}\,.
\]

\begin{proposition}
\label{weig_2}Consider the weight functions
\[
w(t,x,v)=e^{\frac{\alpha\rho(t,x,v)}{2}}\,\ \text{\ \ and}\ \ \ \ \mu
(x,v)=e^{\alpha c(x,v)}\,,
\]
where $\alpha>0$ is sufficiently small with $\alpha\gamma<1/20.$ Then

\noindent\textrm{(i) }For $0\leq j\leq3,$\textrm{ }%
\[
\Vert wh^{(j)}\Vert_{H_{x}^{2}L_{v}^{2}}\lesssim t^{-3+j}\Vert f_{0}%
\Vert_{L^{2}(\mu)},\ \ \ \ \ 0<t\leq1,
\]
and%
\[
\Vert wh^{(j)}\Vert_{H_{x}^{2}L_{v}^{2}}\lesssim\left(  1+t\right)  ^{j}\Vert
f_{0}\Vert_{L^{2}(\mu)},\ \ \ \ \ t>1.
\]

\noindent\textrm{(ii)} \textrm{ }For $1\leq\gamma<3/2$,
\[
\Vert w\mathcal{R}^{(3)}\Vert_{H_{x}^{2}L_{v}^{2}}\lesssim t(1+t)\Vert
f_{0}\Vert_{L^{2}(\mu)}\,,\,\ \ \ t>0,
\]
and for $0<\gamma<1$,
\[
\Vert w\mathcal{R}^{(3)}\Vert_{H_{x}^{2}L_{v}^{2}}\lesssim t(1+t^{4})\Vert
f_{0}\Vert_{L^{2}(\mu)}\,,\,\ \ \ t>0.
\]

\end{proposition}

\begin{proof}
It is similar to Proposition \ref{weig_1} that the weighted energy estimate of
the wave parts $h^{(j)}$ is a consequence of Lemma \ref{second-der} in virtue
of $\rho(t,x,v)$ being non-increasing in $t$ and $\rho(0,x,v)=c(x,v)$.

We shall focus on the weighted energy estimate for the remainder part
$\mathcal{R}^{(3)}$, $t>0$. We want to show that for $1\leq\gamma<3/2$,
\[
\Vert w\mathcal{R}^{(3)}\Vert_{H_{x}^{2}L_{v}^{2}}\lesssim t(1+t)\Vert
f_{0}\Vert_{L^{2}(\mu)}\,,\,\ \ \ t>0,
\]
and for $0<\gamma<1$,
\[
\Vert w\mathcal{R}^{(3)}\Vert_{H_{x}^{2}L_{v}^{2}}\lesssim t(1+t^{4}))\Vert
f_{0}\Vert_{L^{2}(\mu)}\,,\,\ \ \ t>0.
\]

Let $u=w\mathcal{R}^{(3)}=e^{\frac{\alpha\rho}{2}}\mathcal{R}^{(3)}$, and then
$\partial_{x}^{\beta}u,$ where $\beta$ is a multi-index, solves the equation%
\begin{align*}
\partial_{t}\left(  \partial_{x}^{\beta}u\right)   &  =-v\cdot\nabla
_{x}\left(  \partial_{x}^{\beta}u\right)  +\frac{\alpha}{2}(\partial_{t}%
\rho+v\cdot\nabla_{x}\rho)\partial_{x}^{\beta}u+e^{\frac{\alpha\rho}{2}%
}L\left(  e^{-\frac{\alpha\rho}{2}}\partial_{x}^{\beta}u\right) \\
&  \quad+\frac{\alpha}{2}\sum_{\substack{\beta_{1}+\beta_{2}=\beta\\\left\vert
\beta_{1}\right\vert \geq1}}\binom{\beta}{\beta_{1}\,\beta_{2}}(\partial
_{t}\partial_{x}^{\beta_{1}}\rho+v\cdot\nabla_{x}\partial_{x}^{\beta_{1}}%
\rho)\partial_{x}^{\beta_{2}}u\\
&  \quad+\sum_{\substack{\beta_{1}+\beta_{2}+\beta_{3}=\beta\\\left\vert
\beta_{3}\right\vert <\left\vert \beta\right\vert }}\binom{\beta}{\beta
_{1}\,\beta_{2}\,\beta_{3}}(\partial_{x}^{\beta_{1}}e^{\frac{\alpha\rho}{2}%
})L\left(  (\partial_{x}^{\beta_{2}}e^{-\frac{\alpha\rho}{2}})\partial
_{x}^{\beta_{3}}u\right) \\
&  \quad+K\partial_{x}^{\beta}\left(  wh^{(3)}\right)  .
\end{align*}

The energy estimate gives
\[
\begin{aligned} & \quad \frac{1}{2}\frac{d}{dt}\left\Vert \partial_{x}^{\beta}u\right\Vert _{L^{2}}^{2}\\ & =\int_{{\mathbb{R}}^{3}}\left\langle e^{\frac{\alpha \rho}{2}}L\left( e^{-\frac{\alpha\rho}{2}}\partial_{x}^{\beta}u\right) ,\partial_{x}^{\beta}u\right\rangle _{v}dx+\frac{\alpha}{2}\int_{{\mathbb{R}}^{3}}\left\langle (\partial_{t}\rho+v\cdot\nabla_{x}\rho)\partial_{x}^{\beta}u,\partial_{x}^{\beta }u\right\rangle _{v}dx\\ & \quad+\frac{\alpha}{2}\sum_{\substack{\beta_{1}+\beta_{2} =\beta \\ \left\vert \beta_{1}\right\vert \geq1} }\binom{\beta}{\beta_{1} \,\beta_{2}}\int_{{\mathbb{R}}^{3}}\left\langle (\partial_{t}\partial_{x}^{\beta_{1}}\rho+v\cdot \nabla_{x}\partial_{x}^{\beta_{1}}\rho)\partial_{x}^{\beta_{2}}u,\partial _{x}^{\beta}u\right\rangle _{v}dx\\ & \quad +\sum_{\substack{\beta_{1}+\beta_{2}+\beta_{3} =\beta \\ \left\vert \beta_{3}\right\vert <|\beta|} }\binom{\beta}{\beta_{1} \,\beta_{2}\, \beta_{3}}\int_{{\mathbb{R}}^{3}}\left\langle (\partial _{x}^{\beta_{1}}e^{\frac{\alpha\rho}{2}})L\left( (\partial_{x}^{\beta_{2}}e^{-\frac{\alpha\rho}{2}})\partial_{x}^{\beta_{3}}u\right),\partial _{x}^{\beta}u\right\rangle _{v}dx\\ & \quad + \int_{{\mathbb{R}}^{3}}\left\langle K\partial_{x}^{\beta}\left( wh^{(3)}\right) ,\partial_{x}^{\beta}u\right\rangle _{v}dx\\ &\quad := (I_1)+(I_2)+(I_3)+(I_4)+(I_5)\,. \end{aligned}
\]
We shall estimate $(I_{i})\,,\ i=1,\ldots,5,$ term by term. \medskip

For $(I_{1})$, it is easy to see that
\[
\left\langle g,e^{\frac{\alpha\rho}{2}}L\left(  e^{-\frac{\alpha\rho}{2}%
}g\right)  \right\rangle _{v}=\left\langle g,e^{-\frac{\alpha\rho}{2}}L\left(
e^{\frac{\alpha\rho}{2}}g\right)  \right\rangle _{v}=\left\langle
g,Lg\right\rangle _{v}+\frac{\alpha^{2}}{4}\left\langle g^{2},\left\vert
\nabla_{v}\rho\right\vert ^{2}\right\rangle _{v}\,.
\]
In addition, direct calculation gives
\begin{align*}
\nabla_{v}\rho &  =\left[  (\gamma-3)(1-2\chi)\delta(\left\langle
x\right\rangle -Mt)\left\langle v\right\rangle ^{2\gamma-3}+3(\gamma
-3)\left\langle v\right\rangle ^{\gamma}-5(\gamma-3)\left(  \delta
(\left\langle x\right\rangle -Mt)\right)  ^{\frac{\gamma}{3-\gamma}}\right] \\
&  \quad\times\left[  \delta(\left\langle x\right\rangle -Mt)\left\langle
v\right\rangle ^{\gamma-4}\right]  \frac{v}{\left\langle v\right\rangle }%
\chi^{\prime}\\
&  \quad+\left[  (2\gamma-3)\delta(\left\langle x\right\rangle
-Mt)\left\langle v\right\rangle ^{2\gamma-4}\right]  \frac{v}{\left\langle
v\right\rangle }(1-\chi)\chi+3\gamma\left\langle v\right\rangle ^{\gamma
-1}\frac{v}{\left\langle v\right\rangle }\chi\,.
\end{align*}
This implies
\[
|\nabla_{v}\rho|\lesssim\left\langle v\right\rangle ^{\gamma-1}\quad
\text{on}\quad H_{0}\,\cup H_{-}\,,
\]
and
\[
\nabla_{v}\rho=0\quad\text{on}\quad H_{+}\,.
\]
Therefore,
\begin{equation}
\label{formula-I1}\begin{aligned} (I_1) &=\int_{{\mathbb{R}}^{3}}\left\langle e^{\frac{\alpha\rho}{2}}L\left( e^{-\frac{\alpha\rho}{2}}\partial_{x}^{\beta}u\right) ,\partial_{x}^{\beta }u\right\rangle _{v}dx\\ &\leq-\left( \nu_{0}-\frac{\alpha^{2}C}{4}\right) \int_{{\mathbb{R}}^{3}}\left\vert \mathrm{P}_{1}\partial_{x}^{\beta }u\right\vert _{L_{\sigma}^{2}}^{2}dx+\frac{\alpha^{2}C}{4}\int_{H_{0}\cup H_{-}}\left\vert \mathrm{P}_{0}\partial_{x}^{\beta}u\right\vert ^{2}dxdv, \end{aligned}
\end{equation}
for some constant $\nu_{0}>0.$ \medskip

For $(I_{2})$ and $(I_{3})$, we need the estimates of derivatives of
$\rho(t,x,v)$. Direct computation gives%
\begin{align*}
\partial_{t}\rho &  =-\delta M\left\langle v\right\rangle ^{2\gamma-3}\left(
\frac{5\gamma}{3-\gamma}\left[  \delta(\left\langle x\right\rangle
-Mt)\left\langle v\right\rangle ^{\gamma-3}\right]  ^{\frac{2\gamma
-3}{3-\gamma}}\left(  1-\chi\right)  +\chi(1-\chi)\right) \\
&  \quad+\delta M\left(  5\left[  \delta(\left\langle x\right\rangle
-Mt)\left\langle v\right\rangle ^{\gamma-3}\right]  ^{\frac{\gamma}{3-\gamma}%
}-(1-2\chi)\left[  \delta(\left\langle x\right\rangle -Mt)\left\langle
v\right\rangle ^{\gamma-3}\right]  -3\right)  \left\langle v\right\rangle
^{2\gamma-3}\chi^{\prime}\,\leq0,
\end{align*}
(the constants $5$ and $3$ are chosen artificially such that the quantity in
the latter bracket is nonnegative on $H_{0}$) and,
\begin{align*}
\nabla_{x}\rho &  =\delta\left(  \nabla_{x}\left\langle x\right\rangle
\right)  \left\langle v\right\rangle ^{2\gamma-3}\left(  \frac{5\gamma
}{3-\gamma}\left[  \delta(\left\langle x\right\rangle -Mt)\left\langle
v\right\rangle ^{\gamma-3}\right]  ^{\frac{2\gamma-3}{3-\gamma}}\left(
1-\chi\right)  +\chi(1-\chi)\right) \\
&  \quad-\delta\left(  \nabla_{x}\left\langle x\right\rangle \right)  \left(
5\left[  \delta(\left\langle x\right\rangle -Mt)\left\langle v\right\rangle
^{\gamma-3}\right]  ^{\frac{\gamma}{3-\gamma}}-(1-2\chi)\left[  \delta
(\left\langle x\right\rangle -Mt)\left\langle v\right\rangle ^{\gamma
-3}\right]  -3\right)  \left\langle v\right\rangle ^{2\gamma-3}\chi^{\prime
}\,,
\end{align*}
so
\[
\partial_{t}\rho=v\cdot\nabla_{x}\rho=0\quad\text{on}\ H_{-}\,\text{,}%
\]%
\[
|\partial_{t}\rho|\lesssim\delta M\left\langle v\right\rangle ^{2\gamma
-3},\ \ \ \ \ |v\cdot\nabla_{x}\rho|\lesssim\delta\left\langle v\right\rangle
^{2\gamma-2}\ \quad\text{on\ }H_{0}\text{\thinspace,}\ \ \
\]%
\[
\partial_{t}\rho=-\frac{5\delta M\gamma}{3-\gamma}\left[  \delta(\left\langle
x\right\rangle -Mt)\right]  ^{\frac{2\gamma-3}{3-\gamma}},\quad v\cdot
\nabla_{x}\rho=\frac{5\delta\gamma}{3-\gamma}\frac{v\cdot x}{\left\langle
x\right\rangle }\left[  \delta(\left\langle x\right\rangle -Mt)\right]
^{\frac{2\gamma-3}{3-\gamma}}\ \ \ \text{on\ }H_{+}\text{\thinspace.}%
\]
Furthermore, we can also obtain that for $\left\vert \beta_{1}\right\vert
\geq1,$
\[
\partial_{t}\partial_{x}^{\beta_{1}}\rho=\nabla_{x}\partial_{x}^{\beta_{1}%
}\rho=0\ \ \ \ \text{on\ }H_{-}\text{,}%
\]%
\[
\left\vert \partial_{t}\partial_{x}^{\beta_{1}}\rho\right\vert \lesssim
\delta^{2}M\left\langle v\right\rangle ^{\gamma+\left(  \left\vert \beta
_{1}\right\vert +1\right)  \left(  \gamma-3\right)  },\ \ \ \ \left\vert
\nabla_{x}\partial_{x}^{\beta_{1}}\rho\right\vert \lesssim\delta
^{2}\left\langle v\right\rangle ^{\gamma+\left(  \left\vert \beta
_{1}\right\vert +1\right)  \left(  \gamma-3\right)  }\ \ \ \ \text{on\ }%
H_{0}\cup H_{+}\text{.}%
\]
From these, there exist constants $C>0$ and $C^{\prime}>0$ such that
\begin{equation}
\label{formula-I2-1}\begin{aligned} \alpha\left\vert \int_{{\mathbb{R}}^{3}}\left\langle v\cdot\nabla_{x}\rho\partial_{x}^{\beta}u,\partial_{x}^{\beta}u\right\rangle _{v}dx\right\vert & \leq \alpha\delta C\left( \int_{\mathbb{R}^{3}}|\left\langle v\right\rangle ^{\gamma-1}\mathrm{P}_{1}\partial_{x}^{\beta}u|_{L_{v}^{2}}^{2}dx+\int_{H_{0}}|\mathrm{P}_{0}\partial_{x}^{\beta}u|^{2}dxdv\right. \\ & \qquad\left.+\int_{H_{+}}\left[ \delta(\left\langle x\right\rangle -Mt)\right] ^{\frac{2\gamma-3}{3-\gamma}}|\mathrm{P}_{0}\partial_{x}^{\beta}u|^{2}dxdv\right)\, , \end{aligned}
\end{equation}
\begin{equation}
\label{formula-I2-2}\begin{aligned} \alpha\int_{\mathbb{R}^{3}}\left\langle (\partial_{t}\rho)\partial _{x}^{\beta}u,\partial_{x}^{\beta}u\right\rangle _{v}dx & \leq -\alpha\delta MC^{\prime}\int_{H_{+}}\left[ \delta(\left\langle x\right\rangle -Mt)\right] ^{\frac{2\gamma-3}{3-\gamma}}|\mathrm{P}_{0}\partial_{x}^{\beta}u|^{2}dxdv\\ & \quad + \alpha\delta MC\left( \int_{\mathbb{R}^{3}}|\left\langle v\right\rangle ^{\gamma-1}\mathrm{P}_{1}\partial_{x}^{\beta}u|_{L_{v}^{2}}^{2}dx+\int_{H_{0}}|\mathrm{P}_{0}\partial_{x}^{\beta}u|^{2}dxdv\right)\, , \end{aligned}
\end{equation}
and for $\left\vert \beta_{1}\right\vert \geq1,$%
\begin{equation}
\label{formula-I3}\begin{aligned} &\quad \frac{\alpha}{2}\left\vert \int_{\mathbb{R}^{3}}\left\langle (\partial _{t}\partial_{x}^{\beta_{1}}\rho+v\cdot\nabla_{x}\partial_{x}^{\beta_{1}}\rho)\partial_{x}^{\beta_{2}}u,\partial_{x}^{\beta}u\right\rangle _{v}dx\right\vert \\ & \leq\frac{\alpha\delta^{2}M}{2}C\left[ \int_{\mathbb{R}^{3}} \left( \left|\left\langle v\right\rangle ^{\gamma-1}\mathrm{P}_{1}\partial_{x}^{\beta_{2}}u\right|_{L_{v}^{2}}^{2}+|\left\langle v\right\rangle ^{\gamma-1}\mathrm{P}_{1}\partial _{x}^{\beta}u|_{L_{v}^{2}}^{2}\right) dx \right. \\ &\qquad +\int_{H_{0}}\left( \left\vert \mathrm{P}_{0}\partial_{x}^{\beta_{2}}u\right\vert ^{2}+\left\vert \mathrm{P}_{0}\partial_{x}^{\beta}u\right\vert ^{2}\right) dxdv\\ & \qquad\left. +\int_{H_{+}}\left[ \delta(\left\langle x\right\rangle -Mt)\right] ^{\frac{2\gamma-3}{3-\gamma}}\left( \left\vert \mathrm{P}_{0}\partial _{x}^{\beta_{2}}u\right\vert ^{2}+\left\vert \mathrm{P}_{0}\partial_{x}^{\beta}u\right\vert ^{2}\right) dxdv \right] \,. \end{aligned}
\end{equation}

As for $(I_{4})$, $|\beta_{1}|+|\beta_{2}|\geq1$, we have
\begin{equation}
\label{formula-I4}\begin{aligned} &\quad \left| \int_{{\mathbb{R}}^{3}}\left\langle (\partial _{x}^{\beta_{1}}e^{\frac{\alpha\rho}{2}})L\left( (\partial_{x}^{\beta_{2}}e^{-\frac{\alpha\rho}{2}})\partial_{x}^{\beta_{3}}u\right),\partial _{x}^{\beta}u\right\rangle _{v}dx \right|\\ & \leq\frac{\alpha\delta C}{2}\left[ \int_{\mathbb{R}^{3}}\left( \left\vert \mathrm{P}_{1}\partial_{x}^{\beta_{3}}u\right\vert _{L_{\sigma}^{2}}^{2}+\left\vert \mathrm{P}_{1}\partial_{x}^{\beta}u\right\vert _{L_{\sigma }^{2}}^{2}\right) dx+\int_{H_{0}}\left\vert \mathrm{P}_{0}\partial_{x}^{\beta_{3}}u\right\vert ^{2}+\left\vert \mathrm{P}_{0}\partial_{x}^{\beta }u\right\vert ^{2}dxdv\right. \\ & \qquad\left. +\int_{H_{+}}\left[ \delta(\left\langle x\right\rangle -Mt)\right] ^{\frac{2\gamma-3}{3-\gamma}}\left( \left\vert \mathrm{P}_{0}\partial_{x}^{\beta_{3}}u\right\vert ^{2}+\left\vert \mathrm{P}_{0}\partial_{x}^{\beta}u\right\vert ^{2}\right) dxdv\right] . \end{aligned}
\end{equation}

Lastly,
\begin{equation}
(I_{5})\leq\left\vert \int_{{\mathbb{R}}^{3}}\left\langle K\partial_{x}%
^{\beta}\left(  wh^{(3)}\right)  ,\partial_{x}^{\beta}u\right\rangle
_{v}dx\right\vert \lesssim\left\Vert \partial_{x}^{\beta}u\right\Vert _{L^{2}%
}\left\Vert \partial_{x}^{\beta}\left(  wh^{(3)}\right)  \right\Vert _{L^{2}}.
\label{formula-I5}%
\end{equation}

Gathering the terms (\ref{formula-I1})--(\ref{formula-I5}), we find%
\begin{align*}
\frac{d}{dt}\Vert u\Vert_{H_{x}^{2}L_{v}^{2}}^{2}  &  \lesssim\Vert
u\Vert_{H_{x}^{2}L_{v}^{2}}\Vert wh^{(3)}\Vert_{H_{x}^{2}L_{v}^{2}}%
\,+\int_{H_{0}\cup H_{-}}\left(  \left\vert \mathrm{P}_{0}\nabla_{x}%
^{2}u\right\vert ^{2}+\left\vert \mathrm{P}_{0}\nabla_{x}u\right\vert
^{2}+\left\vert \mathrm{P}_{0}u\right\vert ^{2}\right)  dxdv\\
&  \lesssim\Vert u\Vert_{H_{x}^{2}L_{v}^{2}}\Vert wh^{(3)}\Vert_{H_{x}%
^{2}L_{v}^{2}}+\Vert u\Vert_{H_{x}^{2}L_{v}^{2}}\Vert\mathcal{R}^{(3)}%
\Vert_{H_{x}^{2}L_{v}^{2}}\\
&  \lesssim\Vert u\Vert_{H_{x}^{2}L_{v}^{2}}\left(  \Vert h^{(3)}\Vert
_{H_{x}^{2}L_{v}^{2}\left(  \mu\right)  }+\Vert\mathcal{R}^{(3)}\Vert
_{H_{x}^{2}L_{v}^{2}}\right)  ,
\end{align*}
after choosing $\delta$, $\alpha>0$ small and $M$ large enough with
$\alpha\gamma<1/20$. Hence, it follows from Lemmas \ref{initial-sing} and
\ref{second-der} that for $1\leq\gamma<3/2$,
\[
\Vert w\mathcal{R}^{(3)}\Vert_{H_{x}^{2}L_{v}^{2}}=\Vert u\Vert_{H_{x}%
^{2}L_{v}^{2}}\lesssim t(1+t)\left\Vert f_{0}\right\Vert _{L^{2}\left(
\mu\right)  },
\]
and for $0<\gamma<1,$%
\[
\Vert w\mathcal{R}^{(3)}\Vert_{H_{x}^{2}L_{v}^{2}}=\Vert u\Vert_{H_{x}%
^{2}L_{v}^{2}}\lesssim t(1+t^{4})\left\Vert f_{0}\right\Vert _{L^{2}\left(
\mu\right)  }.
\]
This completes the proof of the proposition.
\end{proof}

Observe that for $\left\langle x\right\rangle >2Mt$,
\[
\rho(t,x,v)\gtrsim\left(  \delta(\left\langle x\right\rangle -Mt)\right)
^{\frac{\gamma}{3-\gamma}}.
\]
and
\[
\left\langle x\right\rangle -Mt>\frac{\left\langle x\right\rangle }{3}%
+\frac{Mt}{3}\,.
\]
The Sobolev inequality immediately gets

\begin{theorem}
[Space-like region for $0<\gamma<3/2$]%
\label{space-like region for 0<gamma<3/2} Let $0<\gamma<3/2$ and let $f$ be
the solution to equation (\ref{in.1.c}). Assume that the initial condition
$f_{0}$ has compact support in the $x$ variable and is bounded in $L_{v}%
^{2}(e^{4\alpha\left\langle v\right\rangle ^{\gamma}})$ for $\alpha>0$ small
enough. Then there exists a positive constant $C=C(M)$ such that for
$\left\langle x\right\rangle \geq2Mt$,
\[
\left\vert f\right\vert _{L_{v}^{2}}\lesssim(1+t^{-9/4})e^{-C(\left\langle
x\right\rangle +t)^{\frac{\gamma}{3-\gamma}}}\Vert f_{0}\Vert_{L^{2}%
(e^{4\alpha\left\langle v\right\rangle ^{\gamma}})}\,.
\]

\end{theorem}


\end{document}